%% file: main.tex
\documentclass{article}

\PassOptionsToPackage{authoryear,square,comma}{natbib}
\usepackage[main,final]{neurips_2026}

\usepackage[utf8]{inputenc}
\usepackage[T1]{fontenc}
\usepackage{microtype}
\usepackage{graphicx}
\graphicspath{{./}}
\usepackage{subcaption}
\usepackage{wrapfig}
\usepackage{booktabs}
\usepackage{multirow}
\usepackage{makecell}
\usepackage{hyperref}
\hypersetup{hidelinks,pdfauthor={Hankyeong Ko, Taehoon Kim}}
\usepackage{url}
\usepackage{siunitx}
\usepackage{float}
\usepackage{multicol}
\usepackage{tikz}
\usetikzlibrary{arrows.meta,positioning}
\usepackage{amsmath}
\usepackage{amssymb}
\usepackage{amsfonts}
\usepackage{mathtools}
\usepackage{amsthm}
\usepackage{nicefrac}
\usepackage{xcolor}
\usepackage[capitalize,noabbrev]{cleveref}

\newcommand{\ppmarkalgorithmwrapwidth}{0.50\linewidth}
\floatstyle{ruled}
\newfloat{algorithm}{tbp}{loa}
\floatname{algorithm}{Algorithm}
\newcounter{ALGline}
\newenvironment{algorithmic}[1][0]{%
  \setcounter{ALGline}{0}%
  \begin{list}{}{%
    \setlength{\leftmargin}{2.2em}%
    \setlength{\labelwidth}{1.7em}%
    \setlength{\labelsep}{0.4em}%
    \setlength{\itemsep}{0.2ex}%
    \setlength{\parsep}{0pt}%
    \setlength{\topsep}{0.2ex}%
  }%
}{%
  \end{list}%
}
\newcommand{\STATE}{\refstepcounter{ALGline}\item[\footnotesize\arabic{ALGline}.]}
\newcommand{\FOR}[1]{\STATE \textbf{for} #1 \textbf{do}\begin{list}{}{\setlength{\leftmargin}{1.4em}\setlength{\itemsep}{0.2ex}\setlength{\parsep}{0pt}\setlength{\topsep}{0.2ex}}}
\newcommand{\ENDFOR}{\end{list}}
\newcommand{\COMMENT}[1]{\hfill\(\triangleright\)~#1}

\theoremstyle{plain}
\newtheorem{theorem}{Theorem}[section]

\theoremstyle{definition}

\newtheorem{assumption}[theorem]{Assumption}
\theoremstyle{remark}

\newcommand{\ppmarkdoublefigurewidth}{0.96\textwidth}
\newcommand{\ppmarkalgorithmsize}{\footnotesize}
\newcommand{\ppmarkthreattablewrapwidth}{0.50\linewidth}

\newcommand{\ppmarkablationtablewrapwidth}{0.50\linewidth}

\newcommand{\ppmarksingletablesetup}{%
  \small
  \setlength{\tabcolsep}{5pt}%
  \renewcommand{\arraystretch}{1.05}%
}

\newcommand{\ppmarktighttablesetup}{%
  \footnotesize
  \setlength{\tabcolsep}{2pt}%
  \renewcommand{\arraystretch}{1.00}%
}

\AtBeginDocument{\let\cite\citep}

\title{PP-Mark: Provable and Publicly Verifiable Watermarking for Generative AI}
\author{%
  Hankyeong Ko \qquad Taehoon Kim\thanks{Corresponding author.}\\
  Graduate School of Virtual Convergence\\
  Sogang University\\
  \texttt{hko920920@sogang.ac.kr}\quad\texttt{taehoonkim@sogang.ac.kr}
}

\begin{document}

\maketitle

\begin{abstract}
Generative AI models now produce images indistinguishable from real data, making publicly verifiable provenance essential; however, existing watermarking methods become vulnerable to forgeries and adversarial optimization attacks once their detectors are made public. We propose PP-Mark, a provenance framework that anchors lightweight statistical detection to a zero-knowledge proof of the embedding process, ensuring that forged content cannot pass full verification by construction. We establish formal unforgeability under standard cryptographic assumptions and validate PP-Mark against representative baselines across two latent-diffusion models (SD2.1 and SDXL), demonstrating resistance to black-box imprint forgery, white-box optimization, and regeneration attacks while remaining practical for decentralized verification. Code is available at \url{https://github.com/hko920920/PP-Mark_NeurIPS2026}
\end{abstract}

\section{Introduction}
The rapid proliferation of generative AI has made synthetic content increasingly indistinguishable from human-authored media, facilitating deepfakes, coordinated disinformation, and large-scale digital forgery \cite{chesney_citron2019}. This has created urgent demand for verifiable provenance that can be independently audited by the public \cite{longpre2024dataauth}. Industry-wide initiatives such as C2PA and regulatory frameworks including the EU AI Act now mandate explicit identification of AI-generated content \cite{c2pa_spec_2_3}, but the effectiveness of these standards depends on the robustness of the underlying technical mechanisms.

Digital watermarking is the leading approach for establishing provenance of generative content. Post-hoc methods embed signals in the pixel or frequency domain after generation (e.g., \cite{zhu2018hidden}), but remain vulnerable to regeneration and model-based attacks that can remove the embedded signal \cite{zhao2024removable}. Semantic watermarking mitigates this by injecting structured signals directly into the latent space during the generative process \cite{wen2023treerings,fernandez2023stablesignature,zhang2024zodiac}, exhibiting improved robustness against low-level manipulations.

Despite these advances, existing schemes exhibit a critical vulnerability. First, black-box imprint forgery attacks can extract watermark patterns from a single reference image and transfer them onto forged images by minimizing latent-space distances \cite{muller2025blackbox}. Second, when detectors $D(\cdot)$ are made public, they function as a differentiable or queryable loss, enabling adversaries to optimize forged content to satisfy the verification criterion directly \cite{jiang2023evading}. Public exposure of detection components can further enable surrogate-based optimization attacks that exceed conventional threat models \cite{lin2025crack}.

To mitigate these risks, existing systems rely on centralized verification, keeping $D(\cdot)$ private behind permissioned APIs. While this prevents leakage of the detection rule, it creates a trust bottleneck that precludes independent public auditing. The field thus faces a fundamental tension: private verification sacrifices public auditability, while public verification remains vulnerable to adversarial optimization. Recent work on privacy-preserving verification reflects this same tension \cite{guo2024zeromark}.

We resolve this tension by introducing PP-Mark (Provable Presence Mark), a provenance framework that is both publicly verifiable and resistant to forgery. Our key insight is to reformulate verification from a purely statistical inference problem into a cryptographically grounded proof. Rather than relying solely on a public detection score, which can be optimized against, PP-Mark anchors the verification rule to the specific generation context using zero-knowledge proofs (ZKP), ensuring that successful verification requires cryptographic consistency rather than detector manipulation.

We summarize our contributions as follows.
\begin{itemize}
    \item We introduce a context-bound latent embedding scheme and a zero-knowledge verification framework (based on SP1) that cryptographically anchors the detection rule to the generation context, enabling decentralized public verification without centralized APIs while preventing adversarial optimization for forgery.
    \item We prove formal unforgeability under standard cryptographic assumptions (Theorem~\ref{thm:unforgeability}) and evaluate against representative baselines across two latent-diffusion models.
\end{itemize}

\section{Related Works} 
\subsection{Watermarking Methods and Forgery Vulnerabilities}
Post-hoc methods such as HiDDeN \cite{zhu2018hidden}, StegaStamp \cite{tancik2020stegastamp}, TrustMark \cite{bui2023trustmark}, and InvisMark \cite{xu2024invismark} embed signals after generation for robust detection. Semantic approaches embed during generation: Tree-Ring \cite{wen2023treerings} and Gaussian Shading \cite{yang2024gaussianshading} modulate initial latent noise, while Stable Signature \cite{fernandez2023stablesignature}, DiffuseTrace \cite{lei2024diffusetrace}, and SFW \cite{lee2025iccv_sfw} strengthen resilience through decoder- or process-level injection. In the identification setting, RingID \cite{ci2024ringid} extends Tree-Ring to multi-key separability, WIND \cite{arabi2025wind} proposes a two-stage public detection pipeline, and MaXsive and TraceMark-LDM \cite{luo2025tracemarkldm,mao2025maxsive} push capacity. ISTS~\cite{zhong2024ists} strengthens forgery resistance through instance-specific injection and two-sided detection; PP-Mark pursues a complementary axis by cryptographically anchoring its context-bound latent detection rule to a verifiable proof rather than hardening the detector itself. However, these methods primarily target traceability and remain vulnerable to forgery and replay attacks that exploit exposed detection oracles \cite{jiang2023evading,muller2025blackbox}.

\subsection{Cryptographic Provenance and Verifiable Computation}\label{sec:crypto-binding}
VeriLLM \cite{wang2025verillm} uses sampling-based commitments and randomized openings for publicly verifiable LLM inference; its ``good test'' property, that small random checks detect deviations with non-negligible probability, is analogous to PP-Mark's partial-opening bound (Theorem~\ref{thm:unforgeability}). PP-Mark adapts this commit-open-verify paradigm to image watermarking but additionally operates in zero knowledge and binds the opening set to the published image hash, preventing cross-image proof transfer. Adjacent provable guarantees follow distinct axes: CLUE-Mark~\cite{shehata2024cluemark} provides \emph{undetectability} via the CLWE problem (an adversary without the secret key cannot distinguish watermarked from unwatermarked samples), whereas PP-Mark targets the orthogonal property of \emph{unforgeability}, any accepted image must have been produced by the legitimate generator, and the two are combinable in principle. Unlike C2PA \cite{c2pa_spec_2_3}, which relies on certificate-based metadata signing, PP-Mark proves that a committed embedding trace is consistent with the public binding and anchored to the published image hash, a guarantee signatures alone cannot provide.

  \begin{figure*}[t]
    \centering
    \includegraphics[width=\ppmarkdoublefigurewidth]{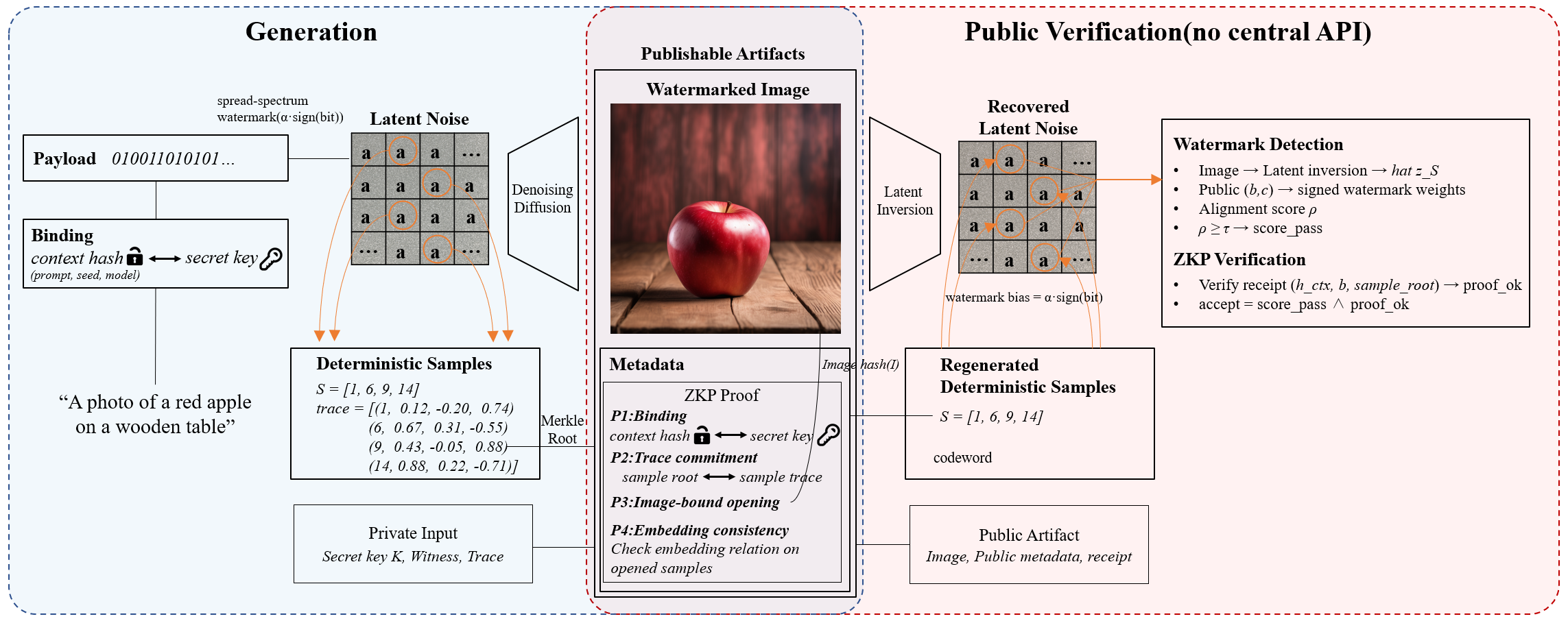}
  \caption{PP-Mark overview. The generator embeds a latent watermark, commits sampled traces to a Merkle root, and produces a ZKP receipt. Verification uses latent inversion to compute a detection score and validates the receipt; acceptance requires both.}
  \label{fig:overview}
  \end{figure*}

\section{Method}
PP-Mark binds a latent-space watermark to the generation context and attaches a standalone SP1 zero-knowledge proof \cite{succinct2024sp1}. Provenance can be verified by any third party using only public artifacts, through a combination of a detection score and cryptographic proof verification; Figure~\ref{fig:overview} illustrates the end-to-end flow. Table~\ref{tab:key-notation} summarizes key notation.

\subsection{Overview and Threat Model}
\noindent
\begin{minipage}[t]{0.49\linewidth}
\vspace{0pt}
\setlength{\parskip}{5.5pt}
Verification combines a calibrated correlation score (FPR controlled by $\tau$) with ZKP receipt validation; concrete schemas and the public/private split are detailed in Sec.~\ref{sec:binding-payload}--\ref{sec:embedding} and Algorithm~\ref{alg:verify}.

We distinguish external forgery from producer non-compliance. An external adversary sees images and public artifacts and may manipulate images, but does not know the producer's binding key $k$; the theorem considers this case for a verifier-pinned producer claim. Non-compliance means that a producer or metadata signer omits embedding, as evaluated in Sec.~\ref{sec:sig-vs-zkp}. Score assumes honest embedding; Accept additionally checks $\mathcal{P}_1$--$\mathcal{P}_4$. The metadata signing key $\mathrm{sk_{sig}}$ used by the signature baselines is distinct from $k$ and may be compromised. Compromise of $k$ is discussed in App.~\ref{Additional Limitations and Discussions}.
\end{minipage}\hfill
\begin{minipage}[t]{0.49\linewidth}
\vspace{0pt}
\captionsetup{type=table}
\caption{Key notation.}
\label{tab:key-notation}
\centering
\setlength{\tabcolsep}{3pt}
\begin{tabular}{@{}>{\raggedright\arraybackslash}p{0.36\linewidth}>{\raggedright\arraybackslash}p{\dimexpr0.64\linewidth-2\tabcolsep\relax}@{}}
\toprule
Symbol & Meaning \\
\midrule
$k$ & Producer's secret binding key. \\
\addlinespace[2pt]
$h_{\mathrm{ctx}}$, $b$, $h_{\mathrm{img}}$ & Context hash; public context--key binding; canonical pixel hash. \\
\addlinespace[2pt]
$S$, $N$ & Sampled positions; $N=|S|$. \\
\addlinespace[2pt]
$K$, $k_{\mathrm{open}}$ & Opened subset $K\subseteq S$; $k_{\mathrm{open}}=|K|$. \\
\addlinespace[2pt]
$z^{(0)}$, $\hat z^{(0)}$ & Embedded / inverted initial latent. \\
\addlinespace[2pt]
$\mathcal T$,\newline $\mathrm{sample\_root}$ & Embedding trace / public Merkle root. \\
\addlinespace[2pt]
$w$, $w_i$ & Private proof witness / expected watermark signal. \\
\bottomrule
\end{tabular}
\par\vspace{\baselineskip}
\end{minipage}
\par

\subsection{Binding and Payload}\label{sec:binding-payload}
PP-Mark binds each watermark to the generation context, the prompt $p$, random seed $s$, model identifier $M$, output resolution $W{\times}H$, and producer's secret key $k$, so that reuse or transfer under a different setting is detectable. We first hash the prompt and then construct a context hash using domain separation:
\begin{equation}
h_{\mathrm{ctx}} = \mathsf{H}(\texttt{ctx\_hash} \,\|\, h_p \,\|\, M \,\|\, s \,\|\, W \,\|\, H),
\quad h_p = \mathsf{H}(p).
\end{equation}

The public binding value is derived from the context hash and the secret key:
\begin{equation}
b = \mathsf{H}(\texttt{ctx\_hash} \,\|\, h_{\mathrm{ctx}} \,\|\, k)
\label{eq:binding_relation}
\end{equation}

Let \(u = (h_{\mathrm{img}}, h_{\mathrm{ctx}}, \mathrm{sample\_root})\). The opening challenge seed is defined as
\begin{equation}
\mathrm{challenge\_seed} = \mathsf{H}(\texttt{ppmark\_opening\_v1} \,\|\, u).
\end{equation}

The proof additionally commits to an opening digest derived from $\mathrm{challenge\_seed}$, corresponding to a deterministic subset of the sampled trace.

Any change to prompt, model, seed, or resolution yields a different $b$, preventing cross-context reuse. From $b$ we derive a fixed-length payload $m = \mathrm{Trunc}(b)$, apply Reed--Solomon encoding to obtain codeword $c$, and expand $c$ into the binary sequence $\mathtt{bits}$ that drives sampling and embedding.

\subsection{Deterministic Sampling and Bit Mapping}
PP-Mark verifies a fixed subset of latent positions, deterministically derived from the public codeword:
\begin{equation}
S = \mathrm{DeterministicSample}(c, N).
\end{equation}

Each spatial coordinate $(x,y)$ maps to a bit index via:
\begin{equation}
\mathrm{idx}(x,y) = \mathsf{H}(b \,\|\, x \,\|\, y \,\|\, 777) \bmod L,
 L = |\mathtt{bits}|.
\end{equation}

A deterministic opening subset $K \subseteq S$ of size $k_{\mathrm{open}}{=}32$ is selected via $\mathrm{challenge\_seed}$. Since $S$ derives from the public $c$, any verifier reproduces the same positions, spreading the payload for robust detection under mild distortions while keeping proof cost bounded and verification publicly auditable.

\subsection{Watermark Embedding}\label{sec:embedding}
For each latent index \(i\), a pseudo-random uniform value \(u_i\) is generated from the
binding value \(b\) and mapped to a raw Gaussian baseline \(\tilde g_i\) via an inverse-CDF lookup
table:
\begin{equation}
u_i = \mathsf{H}(b \,\|\, i{+}1) \bmod 2^{32},
\quad
\tilde g_i = F^{-1}\!\left( \frac{u_i}{2^{32}} \right).
\end{equation}

Each spatial position is indexed as \(i=(x,y)\) on a fixed reference channel. A payload bit
is mapped to a sign \(s_i \in \{-1,+1\}\), and the watermark is injected into the baseline:
\begin{equation}
s_i = 2 \cdot \mathtt{bits}_{\mathrm{idx}(x,y)} - 1,
\quad
z_i = \tilde g_i + \alpha s_i.
\end{equation}

To maintain stable latent variance and ensure that the embedding remains within the
expected distribution of the pre-trained model, we apply the following normalization:
\begin{equation}
z_i^{(0)} = \frac{z_i}{\sqrt{1+\alpha^2}},
\quad
g_i = \frac{\tilde g_i}{\sqrt{1+\alpha^2}},
\quad
\alpha_{\mathrm{eff}} = \frac{\alpha}{\sqrt{1+\alpha^2}}.
\end{equation}
Thus $z_i^{(0)}=g_i+\alpha_{\mathrm{eff}}s_i$; $g_i$ denotes the normalized baseline throughout verification and the committed trace.

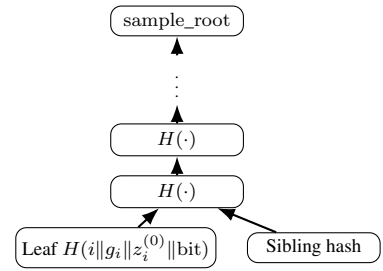
\begin{wrapfigure}{r}{0.43\linewidth}
  \centering
  \vspace{-0.8\baselineskip}
  \footnotesize
  \begin{tikzpicture}[
    scale=0.9,
    transform shape,
    box/.style={draw, rounded corners, align=center, minimum width=1.95cm, minimum height=0.42cm},
    node distance=0.28cm and 0.45cm
  ]
  \node[box] (leafL) {Leaf $H(i \| g_i \| z_i^{(0)} \| \mathrm{bit})$};
  \node[box, right=of leafL] (leafR) {Sibling hash};
  \node[box, above=of leafL, xshift=0.9cm] (h1) {$H(\cdot)$};
  \node[box, above=of h1] (h2) {$H(\cdot)$};
  \node[above=of h2] (dots) {\scriptsize$\vdots$};
  \node[box, above=of dots] (root) {$\mathrm{sample\_root}$};

  \draw[-Latex, thick] (leafL) -- (h1);
  \draw[-Latex, thick] (leafR) -- (h1);
  \draw[-Latex, thick] (h1) -- (h2);
  \draw[-Latex, thick] (h2) -- (dots);
  \draw[-Latex, thick] (dots) -- (root);
  \end{tikzpicture}
  \caption{Merkle authentication path for a sampled trace element used as a proof witness.}
  \label{fig:merkle}
  \vspace{-4\baselineskip}
\end{wrapfigure}

This normalization stabilizes the detection score and reduces global-rescaling sensitivity. For $i \in S$, we record $(i, g_i, z_i^{(0)}, \mathtt{bit}_i)$ in Q18 fixed-point and Merkle-commit (Fig.~\ref{fig:merkle}); the published $\mathrm{sample\_root}$ anchors the SP1 witness, so any tampering invalidates ZKP receipt verification.

\subsection{SP1 Proof Generation}
PP-Mark leverages the SP1 zkVM \cite{succinct2024sp1} to generate a standalone zero-knowledge proof receipt \(\pi\). The proof establishes, in zero knowledge, that the published binding value \(b\) and the committed sampled trace \(\mathcal{T}\), anchored by \(\mathrm{sample\_root}\), are consistent with a valid secret key \(k\) and payload \(\mathtt{bits}\), without revealing these secrets.

We adopt SP1 for engineering practicality: Rust-based tooling and parallelizable proving fit high-throughput generative pipelines. The circuit $\mathcal{C}$ does not re-simulate diffusion; instead, it verifies that the committed trace is consistent with $b$ via the predicate $\mathcal{P}_1 \wedge \mathcal{P}_2 \wedge \mathcal{P}_3 \wedge \mathcal{P}_4$ (formal definition in Sec.~\ref{sec:formal-security}). Since $b$ deterministically governs the sign pattern through $b \to c \to \mathtt{bits} \to s_i$ (Sec.~\ref{sec:binding-payload}--\ref{sec:embedding}), verifying this chain on the opened subset anchors the full detection pipeline at $O(k_{\mathrm{open}})$ cost rather than $O(N)$. End-to-end score verification still requires a compatible initial-latent embedding and recovery interface.

\paragraph{Statement and Witness.}
The SP1 proof $\pi$ attests to a private witness $w=(k,c,\mathcal{T},\mathrm{paths})$ for the public statement $x$ defined in Sec.~\ref{sec:formal-security}. It includes the binding and image/context hashes, trace and opening commitments, challenge seed, $N$, $k_{\mathrm{open}}$, embedding parameters, and LUT hash $h_{\mathrm{LUT}}=\mathsf{H}(F^{-1})$. The verifier checks these public values against the expected statement. The trace $\mathcal{T}=\{(i,g_i,z_i^{(0)},\mathtt{bit}_i)\}_{i\in S}$ stores the normalized baseline and latent in Q18 fixed-point.

\textbf{Circuit Constraints.}
The circuit $\mathcal{C}$ enforces $\mathcal{P}_1 \wedge \mathcal{P}_2 \wedge \mathcal{P}_3 \wedge \mathcal{P}_4$ (binding-key / Merkle / image-bound challenge / deterministic embedding on opened $K$; full definition in Sec.~\ref{sec:formal-security}, with $\mathcal{P}_4$ anchoring traces to the $b$-derived embedding rather than allowing tautological satisfaction).

\subsection{Extraction and Verification}\label{sec:verification}

\begin{wrapfigure}{r}{\ppmarkalgorithmwrapwidth}
\refstepcounter{algorithm}
\label{alg:verify}
\begin{minipage}{\linewidth}
\hrule height .8pt
\vspace{3pt}
\textbf{\strut Algorithm~\thealgorithm: Verify and Extract}
\par\vspace{3pt}
\hrule
\vspace{2pt}
\raggedright
\ppmarkalgorithmsize
\begin{algorithmic}[1]
\STATE Input:
\parbox[t]{0.85\linewidth}{%
image \(I\), metadata \(\mathrm{md}\)\\
(including \(h_{\mathrm{ctx}}, b, \mathrm{sample\_root},\)\\
\(\mathrm{opening\_digest}, k_{\mathrm{open}}\)), receipt \(\pi\)%
}
\STATE \(\hat{z}^{(0)} \leftarrow \text{DDIM\_Inversion}(I)\)
\STATE Reconstruct \(c\) and \(\mathbf{bits}\) from \(b\); regenerate sample set \(S\) from \(c\)
\FOR{each \(i \in S\)}
  \STATE Reconstruct normalized \(g_i\) as in Sec.~\ref{sec:embedding}
  \STATE \(w_i \leftarrow \alpha_{\mathrm{eff}} \cdot s_i\) where \(s_i = 2 \cdot \mathbf{bits}_{\mathrm{idx}(i)} - 1\)
\ENDFOR
\STATE \(\mathrm{score} \leftarrow \frac{|\langle \hat{z}^{(0)}_S,\, w_S \rangle|}{\|w_S\|}\)
\STATE \(h_{\mathrm{img}} \leftarrow \mathsf{SHA\text{-}256}(\texttt{pixels}(I))\) \COMMENT{image enters proof only via hash}
\STATE \(\texttt{proof\_ok} \leftarrow \text{VerifyReceipt}(\pi; \mathrm{md}, h_{\mathrm{img}})\)
\STATE Return \((\mathrm{score} \ge \tau) \land \texttt{proof\_ok}\)
\end{algorithmic}
\vspace{2pt}
\hrule
\end{minipage}
\vspace{-1\baselineskip}
\end{wrapfigure}

Given image $I$, metadata, and receipt $\pi$, the verifier DDIM-inverts $I$ to $\hat{z}^{(0)}$, computes $h_{\mathrm{img}} = \mathsf{SHA\text{-}256}(\texttt{pixels}(I))$, and validates the opening via $\text{VerifyReceipt}(\pi; x)$ (image enters only through $h_{\mathrm{img}}$); the verifier then reconstructs $c$ and $\mathbf{bits}$ from $b$ to derive the expected sign pattern (Algorithm~\ref{alg:verify}).

\begin{assumption}[Statistical Properties of Inversion Residuals]
\label{assum:noise}
Let \(\epsilon_i = \hat{z}^{(0)}_i - z^{(0)}_i\) denote the DDIM inversion residual at latent index \(i\).
We assume that \(\epsilon_i\) follows an approximately zero-mean distribution, i.e.,
\(\mathbb{E}[\epsilon_i] \approx 0\), and has bounded moderate correlation with the
watermark signal \(s_i\).
\end{assumption}

VAE encode--decode roundtrip induces bounded $\epsilon$--$s$ correlation (mean $|\rho|{=}0.30$ on SD\,2.1, $0.54$ on SDXL; App.~\ref{app:assumption-verify}), attenuating the score. Mean watermarked scores exceed $\tau$ by at least $5\times$ in both models; all evaluated watermarked images exceed $\tau$.

Under Assumption~\ref{assum:noise}, for each sampled index \(i \in S\) the recovered
latent value is \(\hat{z}^{(0)}_i = z^{(0)}_i + \epsilon_i\), where \(z^{(0)}_i = g_i + \alpha_{\mathrm{eff}}\, s_i\). Defining the expected watermark signal \(w_i = \alpha_{\mathrm{eff}}\, s_i\), the inner product
\(\langle \hat{z}^{(0)}_S, w_S \rangle = \alpha_{\mathrm{eff}}^2 \|s_S\|^2 + \langle g_S, w_S\rangle + \langle \epsilon_S, w_S\rangle\).
Because \(g\) and \(s\) are derived from independent hash chains, \(\langle g_S, w_S\rangle \approx 0\) in expectation.
\label{eq:residual}

PP-Mark uses $N{=}1000$ samples, balancing score-variance reduction (via averaging) against SP1 proving cost (App.~\ref{app:sample-size} ablates $N$). The detection score is defined as the normalized correlation over the sample set $S$:
\begin{equation}
\mathrm{score}
= \frac{|\langle \hat{z}^{(0)}_S,\, w_S \rangle|}{\|w_S\|}.
\end{equation}

An image is accepted as authentic if and only if it satisfies both the statistical and cryptographic verification criteria:

\begin{equation}
\text{Accept}(I,\pi)\iff (\mathrm{score}\ge\tau)\land(\text{VerifyReceipt}(\pi; x)=1).
\end{equation}

\subsection{Synchronization and Security Rationale}\label{sec:synchronization}
To handle transformation-induced misalignment, PP-Mark employs an alignment search over latent positions. Deterministic sampling at $N{=}1000$ keeps proving cost tractable while Hoeffding concentration~\cite{hoeffding1963} preserves the statistical reliability required for stable detection.

Binding $b$ ties the watermark to the full generation context (any prompt or seed change yields a distinct commitment); anchoring score detection to ZKP receipt $\pi$ ensures verification fails for score-optimized forgeries lacking a consistent latent trace.

\raggedbottom
\begin{wraptable}{r}{\ppmarkthreattablewrapwidth}
\vspace*{-0.6\baselineskip}               
\setlength{\abovecaptionskip}{0pt}%
\caption{Threat coverage by verification mode. In score mode the verifier trusts that the generator embedded honestly: a malicious generator could fabricate a binding $b'$ whose induced sign pattern exceeds $\tau$ on an unwatermarked image (Table~\ref{tab:sig-vs-zkp}). Accept mode blocks this via ZKP verification of the $b{\to}\mathtt{bits}{\to}z$ derivation chain.}
\label{tab:threat-coverage}
{\footnotesize
\setlength{\tabcolsep}{2pt}%
\renewcommand{\arraystretch}{1.0}%
\begin{tabular}{@{}>{\raggedright\arraybackslash}p{0.68\linewidth}cc@{}}
\toprule
Threat model & Score & Accept \\
\midrule
External adversary (unknown $k$) & \checkmark & \checkmark \\
Generator fraud (no embedding) & $\times$ & \checkmark \\
Third-party audit (no $k$ disclosure) & $\times$ & \checkmark \\
\bottomrule
\end{tabular}}
\vspace{-0.8\baselineskip}
\end{wraptable}

The statistical and cryptographic layers share a common trust anchor: the sign pattern $s_i$ that drives the score is deterministically keyed by $b$. Without assurance that $b$ was legitimately derived, an adaptive adversary could fabricate a binding $b'$ whose induced sign pattern maximizes the score for an arbitrary target image; the ZKP forecloses this fabricated-binding attack by certifying that $b$ originates from a valid secret key and generation context (Theorem~\ref{thm:unforgeability}). Conversely, the score layer provides graceful degradation under lossy transforms (App.~\ref{Image Transformations}) that the proof alone cannot. Image-binding is achieved exclusively through $h_{\mathrm{img}}$ in the challenge seed ($\mathcal{P}_3$, Sec.~\ref{sec:formal-security}), which prevents cross-image proof transfer; the two layers thus address complementary threat surfaces through the same cryptographic foundation.

This complementarity motivates two operational modes. \textbf{PP-Mark(score)} checks $\mathrm{score}\geq\tau$ for routine screening (0.72\,s) and tolerates common lossy transforms. \textbf{PP-Mark(accept)} additionally checks an image-bound receipt for high-assurance auditing. 17.0\,s projected receipt issuance at $N=1000$ (Table~\ref{tab:proof-backend-comparison}); a receipt is issued once and can be cached. Table~\ref{tab:threat-coverage} and Sec.~\ref{sec:sig-vs-zkp} distinguish the threat models. Score and Accept provide explicitly different assurance levels.

\paragraph{Image hashing for $h_{\mathrm{img}}$.}
$h_{\mathrm{img}}$ is SHA-256 over canonicalized pixels; exact-image binding refers to these pixels, not file-container bytes. A changed canonical image needs a fresh receipt only if it must itself receive Accept. Score screening of lossy derivatives requires no re-proof (App.~\ref{app:verification-design}).

\subsection{Formal Security Guarantee}\label{sec:formal-security}

\begin{assumption}\label{asm:crypto}
SHA-256 is collision- and preimage-resistant and is modeled as a random oracle for challenge sampling. The receipt proof is assumed knowledge-sound for the declared relation with security parameter $\lambda$. The verifier pins the producer claim $(b,h_{\mathrm{ctx}})$ and checks the receipt's public values against the expected image and protocol parameters.
\end{assumption}

\paragraph{Composite verification predicate.}
Let \(x\) denote the public statement and \(w\) the private witness:
\begin{align*}
x &= (b,h_{\mathrm{ctx}},h_{\mathrm{img}},\mathrm{sample\_root},\mathrm{opening\_digest}, \\
  &\qquad \mathrm{challenge\_seed},N,k_{\mathrm{open}},\alpha,\alpha_{\mathrm{eff}},h_{\mathrm{LUT}}), \\
w &= (k, c, \mathcal{T}, \mathrm{paths}), \qquad
\mathcal{T}=\{(i,g_i,z_i^{(0)},\mathtt{bit}_i)\}_{i \in S}.
\end{align*}
The SP1 circuit \(\mathcal{C}\) enforces the composite predicate
\(\mathcal{P}(x, w) = \bigwedge_{j=1}^{4} \mathcal{P}_j(x, w)\):
\begin{align*}
\mathcal{P}_1&: b = \mathsf{H}(\texttt{ctx\_hash} \,\|\, h_{\mathrm{ctx}} \,\|\, k), \\
\mathcal{P}_2&: \forall i \in S,\; \mathrm{FoldMerkle}(\mathrm{leaf}_i, \mathrm{paths}_i) = \mathrm{sample\_root}, \displaybreak[0] \\
\mathcal{P}_3&: \mathrm{challenge\_seed}=\mathsf{H}(\texttt{ppmark\_opening\_v1}\|h_{\mathrm{img}}\|h_{\mathrm{ctx}}\|\mathrm{sample\_root}), \\
&\quad \mathrm{opening\_digest} = \mathsf{H}\!\left( \{\mathrm{leaf\_tuple}_i\}_{i \in K} \right), \\
&\text{where } K = \mathrm{DeriveOpen}(\mathrm{challenge\_seed}, k_{\mathrm{open}}, N), \displaybreak[0] \\
\mathcal{P}_4&: \forall i \in K,\;
  \mathtt{bit}_i = \mathtt{bits}_{\mathrm{idx}(i)}, \quad
  g_i = \mathsf{Q}_{18}\!\left(\frac{F^{-1}(u_i/2^{32})}{\sqrt{1+\alpha^2}}\right), \\
  &\qquad\qquad\; z_i^{(0)} = \mathsf{Q}_{18}\!\left(g_i + \alpha_{\mathrm{eff}} \cdot (2\mathtt{bit}_i - 1)\right).
\end{align*}
\flushbottom

Here $\mathtt{bits}=\mathrm{BitExpand}(\mathrm{RS}(\mathrm{Trunc}(b)))$, $u_i$ follows Sec.~\ref{sec:embedding}, and $\mathsf{Q}_{18}$ denotes the declared fixed-point rounding convention. The embedding parameters and LUT are fixed by the verified configuration. Clauses $\mathcal{P}_1$--$\mathcal{P}_4$ check key binding, trace commitment, image-bound opening, and deterministic embedding consistency, respectively. The sampling bound below requires the candidate image and trace to be fixed before its opening challenge is queried.

\begin{theorem}[Unforgeability under composite predicate and partial opening]\label{thm:unforgeability}
Under Assumption~\ref{asm:crypto}, fix a target producer claim $(b,h_{\mathrm{ctx}})$. Consider one candidate whose canonical image and trace are fixed before its first opening-challenge query, with $K$ sampled uniformly without replacement. For a probabilistic polynomial-time adversary $\mathcal{A}$ without $k$ whose committed trace has at least $q$ embedding violations,
\begin{align*}
\Pr[&\mathrm{PP\text{-}Mark}_{\mathrm{accept}}(I^*, \mathrm{md}^*, \pi^*) = 1 \;\wedge\; (I^*, \pi^*) \notin \mathcal{G}] \\
&\leq \; \mathrm{Adv}^{\mathrm{hash}}_{\mathsf{H}}(\mathcal{A}) + \mathrm{Adv}^{\mathrm{sound}}_{\mathrm{SP1}}(\mathcal{A}) + \binom{N-q}{k_{\mathrm{open}}} \big/ \binom{N}{k_{\mathrm{open}}},
\end{align*}
where $\mathcal{G}$ contains honest image-proof pairs for the target claim, $\mathrm{Adv}^{\mathrm{hash}}$ aggregates the required hash failures, and $\mathrm{Adv}^{\mathrm{sound}}$ is the assumed knowledge-soundness error. The last term is at most $(1-q/N)^{k_{\mathrm{open}}}$; the bound is per fixed candidate.
\end{theorem}

\begin{proof}[Proof sketch; full reduction in Appendix~\ref{app:full-proof}]
Acceptance requires both $\mathrm{score}\geq\tau$ and a receipt matching the expected public statement. Under the assumed knowledge-soundness, a valid receipt yields a witness satisfying $\mathcal{P}_1$--$\mathcal{P}_4$, except with error $\mathrm{Adv}^{\mathrm{sound}}_{\mathrm{SP1}}$.

\emph{Binding and commitment.} For the pinned claim, recovering a valid secret-key witness or changing a committed trace without changing its root requires a preimage or collision failure, included in $\mathrm{Adv}^{\mathrm{hash}}_{\mathsf{H}}$.

\emph{Partial opening.} Under the fixed-candidate and uniform-sampling conditions, $K$ misses the $q$ violating positions with probability $\binom{N-q}{k_{\mathrm{open}}}/\binom{N}{k_{\mathrm{open}}}$. Otherwise $\mathcal{P}_4$ rejects. The image-derived challenge alone does not establish these sampling conditions for adaptively selected candidates.

\emph{Image binding.} A receipt for different canonical pixels fails the public-statement comparison unless their hashes collide. This argument does not require different challenge seeds to produce different subsets $K$.

Summing the three failure modes yields the stated bound. 
\end{proof}

\paragraph{Parameter choice and remaining guarantees.} At $N{=}1000$, $k_{\mathrm{open}}{=}32$, and $q{=}125$, the exact miss probability is $0.01297$ (98.7\% spot-check detection); at $q{=}500$ it is at most $2^{-32}$. The bound weakens for small $q$, and clean-FPR calibration alone does not bound adversarial FAR. Appendix~\ref{app:full-proof} gives the scope and fixed-point details.

\section{Experiments}
\paragraph{Experimental Setup.}
We use Stable Diffusion v2.1 base with 50 DDIM generation steps and guidance scale 7.5. Generation steps are distinct from inversion steps: the transformation evaluation uses FP32 DDIM-100 inversion and alignment (App.~\ref{Image Transformations}). PP-Mark uses $\alpha=4.0$, $N=1000$, and SHA-256 sampling. Prompts are drawn from the Stable-Diffusion-Prompts dataset; prompt--seed lists and frozen calibration settings accompany the reproduction materials. All methods use native detection statistics calibrated to 1\% FPR on clean SD2.1 images. SDXL results are in App.~\ref{app:sdxl}; BRISQUE, KID, and CLIP comparisons are in App.~\ref{app:quality} (Table~\ref{tab:quality-comparison}).

\paragraph{Runtime.}
PP-Mark$_{\mathrm{score}}$ retains the submitted 0.89\,s embedding and 0.72\,s verification times. SP1 v6.3.1 at $N=1000$, reports 17.0\,s projected receipt issuance and a 2.81\,MB artifact. Tables~\ref{tab:runtime}--\ref{tab:proof-backend-comparison} distinguish the timing boundaries and submitted versus updated configurations.

\subsection{Verification Ablation}\label{sec:sig-vs-zkp}

\begin{wraptable}[10]{r}{\ppmarkablationtablewrapwidth}
\vspace{-1.3\baselineskip}
\setlength{\abovecaptionskip}{0pt}%
\setlength{\belowcaptionskip}{7pt}%
\caption{Forgery Acceptance Rate (FAR) under two attacks across four verification architectures ($n=100$, FPR\,=\,1\%).}
\label{tab:sig-vs-zkp}
\centering
{\footnotesize
\setlength{\tabcolsep}{1.5pt}%
\renewcommand{\arraystretch}{1.0}%
\begin{tabular}{@{}>{\raggedright\arraybackslash}p{0.31\linewidth}cccc@{}}
\toprule
Attack & Sig-only & Sig+Score & \makecell{PP-Mark\\(score)} & \makecell{PP-Mark\\(accept)} \\
\midrule
Forged binding & 1.00 & 1.00 & ${\leq}$0.01$^{\dagger}$ & \textbf{0.00} \\
Compliance bypass & 1.00 & 0.00 & ${\leq}$0.01$^{\dagger}$ & \textbf{0.00} \\
\bottomrule
\end{tabular}
}

\vspace{-3pt}

\begin{flushleft}
{\scriptsize $^{\dagger}$Upper-bounded by FPR; 0/100 observed. Binding $b$ is cryptographically fixed by the secret key.\par}
\end{flushleft}  

\vspace{-1\baselineskip}
\end{wraptable}

To validate the necessity of zero-knowledge verification (Table~\ref{tab:threat-coverage}), we compare four verification architectures under two targeted attacks on $n=100$ watermarked images at FPR\,=\,1\%:

\smallskip
\noindent\textbf{Forged binding.} An adversary who controls only the metadata signing key $\mathrm{sk_{sig}}$, a standard digital-signature key used by the Sig-only and Sig+Score baselines to authenticate metadata, independent of and unrelated to PP-Mark's binding key $k$, fabricates a binding $b'$ to maximize the detection score for a target image. We evaluate $10{,}000$ random bindings per image (sufficient to expect ${\sim}70$ scores exceeding $\tau$ under the null per-trial probability of ${\approx}0.7\%$) and report the fraction for which the best score exceeds $\tau$.

\smallskip
\noindent\textbf{Compliance bypass.} An adversary generates a clean (unwatermarked) image and signs legitimate-looking metadata without embedding any watermark signal.

Table~\ref{tab:sig-vs-zkp} reports the results. Sig-only fails both attacks; Sig+Score blocks compliance bypass but not forged binding because random search over $10{,}000$ bindings consistently finds one whose induced sign pattern exceeds $\tau$ (mean best score $3.88 \gg \tau{=}2.47$). PP-Mark(score) resists both because the binding is cryptographically fixed by the secret key $k$; PP-Mark(accept) additionally rejects any fabricated binding via ZKP, confirming Theorem~\ref{thm:unforgeability} empirically.

For exact canonical-image integrity under a trusted signer, an image-hash signature is simpler and sufficient; App.~\ref{app:signature-controls} isolates PP-Mark's added embedding-compliance evidence and transform-tolerant Score.

\raggedbottom

\subsection{Imprint Forgery Attack}\label{sec:imprint-forgery}
Imprint forgery is a strong black-box and method-agnostic threat: from a
single watermarked reference image, an attacker attempts to transfer the watermark imprint onto an arbitrary cover so the verifier attributes it to the same source
\cite{muller2025blackbox}. This captures a realistic adversary who does not know the target watermark design but can adapt a proxy generator, making it a practical strong attack. We use SD2.1 for both the target and the proxy in the imprint-forgery attack. To evaluate this threat fairly, we construct a fixed prompt pool (200 prompts; 100/100 calibration/evaluation) and prepare a cover pool of 100 images as targets for attribution transfer. We fix 100 cover--reference pairs (one reference per cover) and reuse the same pairs across all methods to avoid pairing bias.

\paragraph{Direct sanity check.}
We first apply the standard imprint-forgery optimizer ($n{=}100$ per step, $T \in \{50, 100, 150\}$) to Tree-Ring and Gaussian Shading \cite{yang2024gaussianshading} to confirm the attack pipeline functions correctly; both are highly vulnerable (TR FAR@1\%FPR: 0.92--0.97; GS: 1.00; full breakdown in Table~\ref{tab:direct-vulnerability}, App.~\ref{app:imprint-forgery-details}).

\paragraph{Transfer (cross-method).}
\begin{wrapfigure}[13]{r}{0.56\linewidth}
    \vspace{-0.6\baselineskip}
    \centering
    \includegraphics[width=\linewidth]{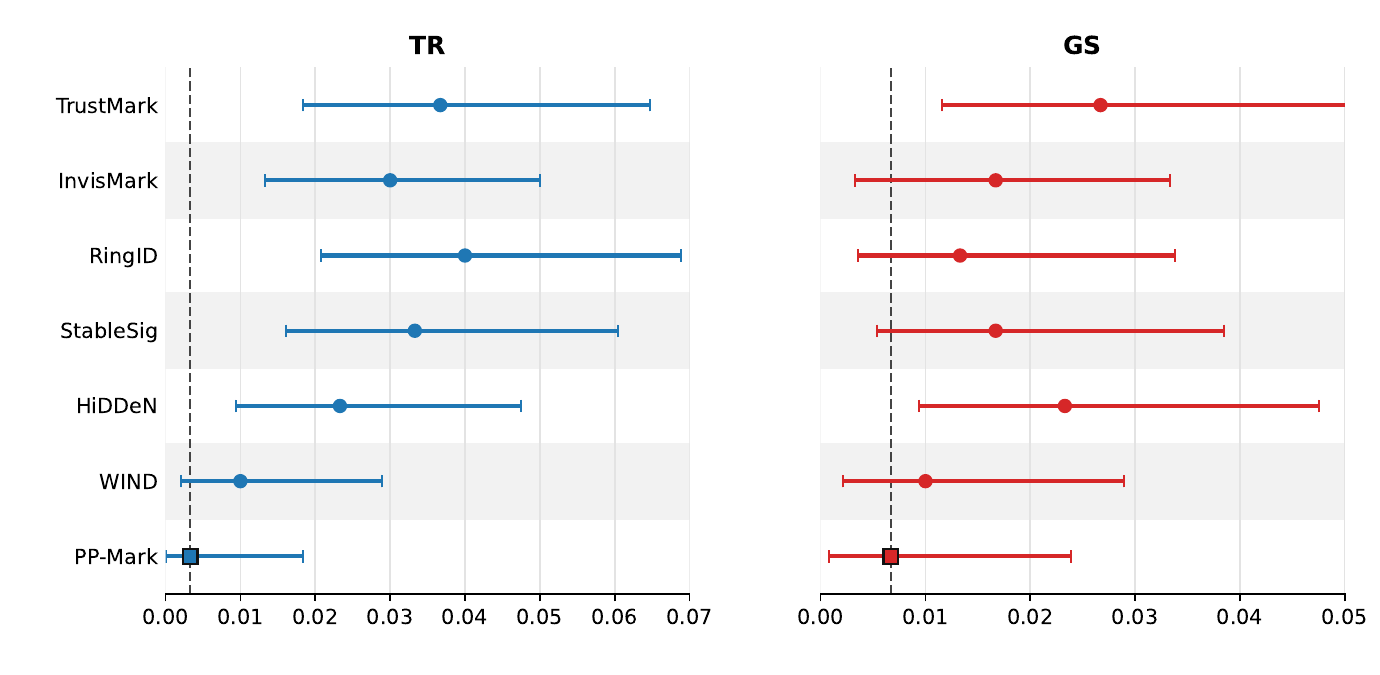}
    \caption{Cross-method transfer at FPR=1\%. Error bars denote 95\% CI; the vertical line indicates PP-Mark$_{\mathrm{score}}$.}
    \label{fig:transfer-forest}
    \vspace{-1.4\baselineskip}
\end{wrapfigure}
We extend the pipeline cross-method: imprints extracted from a TR- or GS-watermarked reference are optimized onto each fixed cover, with all detectors evaluated at thresholds calibrated to 1\% FPR on clean images. Across step budgets $T \in \{50, 100, 150\}$, PP-Mark$_{\mathrm{score}}$ remains at 0.00--0.01, while TrustMark, InvisMark, and other baselines (RingID, StableSig, HiDDeN) exceed the 1\% calibration floor (Table~\ref{tab:transfer-steps}; bit-accuracy calibration for binary-only detectors in App.~\ref{app:imprint-forgery-details}).

\begin{table}[H]
\setlength{\abovecaptionskip}{0pt}%
\setlength{\belowcaptionskip}{2pt}%
    \caption{Transfer summary over steps.
    FAR@1\%FPR (targeted verify) with 95\% CI.
    All methods evaluated at FPR=1\% (clean-calibrated).
    $\Delta$\% is computed as $(k-k_{\mathrm{PP\text{-}Mark_{score}}})/k_{\mathrm{PP\text{-}Mark_{score}}} \times 100$, where $k_{\mathrm{PP\text{-}Mark_{score}}}=1$ (TR) and $k_{\mathrm{PP\text{-}Mark_{score}}}=2$ (GS). Fisher $p$ is from a one-sided Fisher exact test vs PP-Mark$_{\mathrm{score}}$ (H1: method $>$ PP-Mark$_{\mathrm{score}}$).}
    \label{tab:transfer-summary}
    \centering
    {\ppmarktighttablesetup
    \begin{tabular}{
    l
    S[table-format=1.4, round-mode=places, round-precision=4, round-pad=true] c c c
    S[table-format=1.4, round-mode=places, round-precision=4, round-pad=true] c c c
    }
      \toprule
      & \multicolumn{4}{c}{\textbf{TR}} & \multicolumn{4}{c}{\textbf{GS}} \\
      \cmidrule(lr){2-5}\cmidrule(lr){6-9}
      \textbf{Method} & {\textbf{FAR@1\%FPR}} & \textbf{95\% CI} & \textbf{$\Delta$\%} &
\textbf{Fisher $p$}
      & {\textbf{FAR@1\%FPR}} & \textbf{95\% CI} & \textbf{$\Delta$\%} & \textbf{Fisher $p$} \\
      \midrule
      TrustMark & 0.0367 & 0.0184--0.0647 & +1000.0 & 0.0029 & 0.0267 & 0.0116--0.0519 & +300.0 & 0.0532 \\
      InvisMark & 0.0300 & 0.0138--0.0562 & +800.0 & 0.0102 & 0.0167 & 0.0054--0.0385 & +150.0 & 0.2252 \\
      RingID & 0.0400 & 0.0208--0.0688 & +1100.0 & 0.0016 & 0.0133 & 0.0036--0.0338 & +100.0 &
0.3430 \\
      StableSig & 0.0333 & 0.0161--0.0604 & +900.0 & 0.0055 & 0.0167 & 0.0054--0.0385 & +150.0 &
0.2252 \\
      HiDDeN & 0.0233 & 0.0094--0.0475 & +600.0 & 0.0343 & 0.0233 & 0.0094--0.0475 & +250.0 & 0.0882 \\
      WIND & 0.0100 & 0.0021--0.0289 & +200.0 & 0.3119 & 0.0100 & 0.0021--0.0289 & +50.0 & 0.5000 \\
      PP-Mark$_{\mathrm{score}}$ & 0.0033 & 0.0001--0.0184 & -- & -- & 0.0067 & 0.0008--0.0239 & -- & -- \\
      PP-Mark$_{\mathrm{accept}}$ & 0.0000 & -- & -- & -- & 0.0000 & -- & -- & -- \\
      \bottomrule
    \end{tabular}
    }
    \vskip -0.1in
  \end{table}

\paragraph{Summary.}
Table~\ref{tab:transfer-summary} summarizes aggregated transfer acceptance. PP-Mark$_{\mathrm{score}}$ remains at or below the 1\% calibration floor (TR: 0.0033, GS: 0.0067), so the imprint-forgery attack fails to move PP-Mark beyond the clean-image noise floor; baselines exhibit substantially higher TR-sourced transfer.

The $\Delta\%$ column is a descriptive summary computed against PP-Mark$_{\mathrm{score}}$'s small denominators; the absolute FARs and 95\% CIs provide the primary basis for comparison and consistently rank PP-Mark$_{\mathrm{score}}$ as the least susceptible method.

One-sided Fisher exact tests against PP-Mark$_{\mathrm{score}}$ confirm this ordering: TR-sourced differences for RingID, Stable Signature, and HiDDeN are significant ($p=0.0016$, $0.0055$, $0.0343$); GS-sourced differences are not significant at $\alpha{=}0.05$, but the point estimates and CIs preserve the same ordering with PP-Mark$_{\mathrm{score}}$ as least vulnerable.

\subsection{White-Box Forgery Attack}\label{sec:white-box}
\vspace{-1\baselineskip}
\begin{wraptable}[16]{r}{0.5\linewidth}
  \vspace{-2.1\baselineskip}
  \centering
  \setlength{\abovecaptionskip}{0pt}%
  \setlength{\belowcaptionskip}{6pt}%
  \caption{White-Box forgery outcomes across methods. Avg.\ Steps to FA reports, among forgery acceptances, the mean optimization step at which $\mathrm{score}\ge\tau$ first occurs. PSNR is averaged against the original cover image.}
  \label{tab:white-box-summary}
  {\ppmarktighttablesetup
  \setlength{\tabcolsep}{5.5pt}
    \begingroup
    \hypersetup{hidelinks}
    \begin{tabular}{
      l
      S[table-format=1.2, round-mode=places, round-precision=2, round-pad=true]
      S[table-format=3.2, round-mode=places, round-precision=2, round-pad=true]
      S[table-format=2.2, round-mode=places, round-precision=2, round-pad=true]
    }
      \toprule
      \textbf{Method}
      & {\makecell{\textbf{FAR}\\ \textbf{@1\%FPR}}}
      & {\makecell{\textbf{Avg.\ }\\ \textbf{Steps to FA}}}
      & {\makecell{\textbf{Avg.\ }\\ \textbf{PSNR $\uparrow$}}} \\
      \midrule
      TrustMark                 & 1.00 & 2.41  & 75.73 \\
      InvisMark                 & 1.00 & 3.07  & 72.67 \\
      StableSig                 & 1.00 & 5.40  & 65.91 \\
      HiDDeN                    & 1.00 & 8.08  & 62.96 \\
      RingID                    & 1.00 & 2.26  & 74.51 \\
      WIND                      & 1.00 & 22.20 & 55.76 \\ 
      PP-Mark$_{\mathrm{score}}$  & 0.76 & 17.42 & 67.69 \\
      PP-Mark$_{\mathrm{accept}}$ & 0.00 & \multicolumn{1}{c}{--} & \multicolumn{1}{c}{--} \\
      \bottomrule
    \end{tabular}
    \endgroup
  }
  \vspace{0.4\baselineskip}
\end{wraptable}
We consider a strongest case white-box adversary with full access to each detector and
gradients, optimizing inputs to maximize the detection score under an $\ell_\infty$ budget $\epsilon=2/255$ using $T=150$ PGD steps. For each method, forgery acceptance is defined by $\mathrm{score}\ge\tau$, where $\tau$ is calibrated at 1\% FPR.

Table~\ref{tab:white-box-summary} shows all baselines reach FAR@1\%FPR=1.00, while PP-Mark$_{\mathrm{score}}$ remains lower at 0.76. The smaller mean step than WIND despite lower FAR reflects PP-Mark's keyed image-conditioned objective (variable per-image progress) versus WIND's shared fixed-pattern objective (uniform but later acceptance); cumulative FAR trajectories in Fig.~\ref{fig:wb_forge_far_curves}.

\begin{wrapfigure}[12]{r}{0.5\linewidth}
    \vspace{-2.5\baselineskip}
    \centering
    \includegraphics[width=\linewidth]{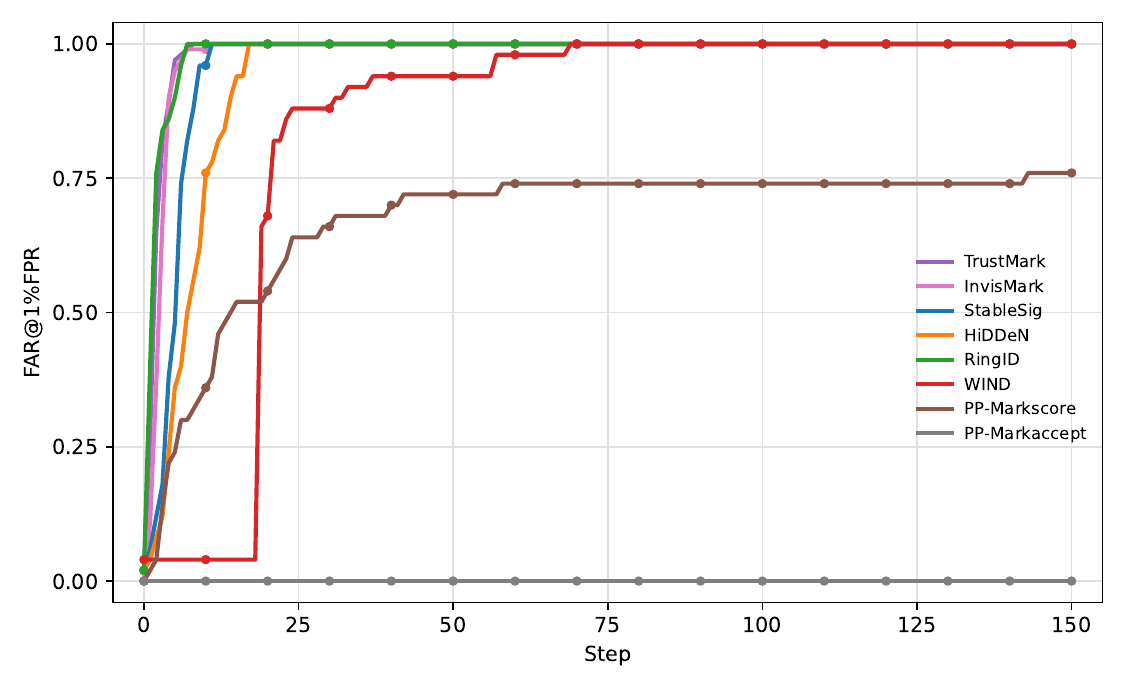}
    \caption{Cumulative FAR@1\%FPR vs. step (same attack budget across all methods).}
    \label{fig:wb_forge_far_curves}
    \vspace{0\baselineskip} 
\end{wrapfigure}

All methods are attacked via their differentiable pipelines: TrustMark, InvisMark, Stable Signature, and HiDDeN use differentiable decoders; RingID and PP-Mark$_{\mathrm{score}}$ use differentiable inversion-based scores. For WIND, we optimize a differentiable Fourier-distance surrogate while evaluating against the actual detector.

For PP-Mark$_{\mathrm{accept}}$, watermark acceptance requires proof verification. The
attacked image changes the image hash $h_{\mathrm{img}}$, which changes 
$\mathrm{challenge\_seed}$ and thus the opening positions and the recomputed 
$\mathrm{opening\_digest}$. \mbox{As a result}, the verifier's opening check fails and the proof verification rejects, keeping PP-Mark$_{\mathrm{accept}}$ at FAR=0.

\subsection{Removal Attacks}
PP-Mark survives standard watermark erasure across three attack families at FPR${\leq}$1\%. \textbf{Regeneration} (diffusion re-synthesis \cite{zhao2024removable}): PP-Mark TPR 0.98 vs StableSig/HiDDeN 0.01--0.06, InvisMark ${\le}0.05$, TrustMark 0.85${\to}$0.39 across noise steps; latent-space embedding survives where pixel-space methods are washed out. \textbf{Steganalysis} ($K$ reference samples): PP-Mark TPR 0.864--0.963 vs Tree-Ring 0.02--0.038, StableSig 0.17--0.20. \textbf{Pixel/geometric transformations}: JPEG/blur/noise/brightness 0.86--1.00; random crop SD\,2.1: 0.87, SDXL: 0.74; rotation $75^\circ$ via rotation-sync (SD\,2.1: 0.85, SDXL: 0.97). Full configurations and per-attack results are in Apps.~\ref{sec:Regeneration Attack}, \ref{Steganalysis Attack}, and~\ref{Image Transformations} (Table~\ref{tab:image-transforms}).

\flushbottom
\section{Conclusion}
PP-Mark combines context-bound statistical detection with public proof verification of a committed embedding relation. Score supports transformation-tolerant screening; Accept additionally checks an image-bound receipt under the stated trust and sampling assumptions. Experiments on SD2.1 and SDXL demonstrate the resulting forgery resistance in the evaluated settings. The proof checks embedding consistency rather than re-executing the generator. Limitations (proof overhead, key lifecycle, geometric robustness, and generative-model scope) are discussed in App.~\ref{Additional Limitations and Discussions}.

\section{Broader Impacts}\label{app:broader-impacts}
This work provides a decentralized, verifiable mechanism for AI image provenance, helping reduce misattribution and synthetic media misuse. By removing reliance on centralized APIs and addressing forgery threats, PP-Mark enables transparent auditing of provenance claims under clear operational assumptions. At the same time, stronger provenance and watermark-verification tools may increase the traceability of generated media in some deployment contexts, which could be undesirable where anonymity or sensitive expression is important. We therefore view PP-Mark as a supporting provenance tool that should be deployed with appropriate transparency and user-facing safeguards.

\section*{Acknowledgements and Disclosures}
This research was supported by the Institute for Information \& Communications Technology Planning \& Evaluation (IITP) grant funded by the Ministry of Science and ICT (MSIT), Republic of Korea, through the Graduate School of Metaverse Convergence Program (RS-2022-00156318); by the Ministry of Culture, Sports and Tourism (MCST), Republic of Korea, through the Korea Creative Content Agency (KOCCA) under the Culture, Sports and Tourism R\&D Program (RS-2023-00219237 and RS-2026-25506280); and by the National Research Foundation of Korea (NRF) grant funded by the Ministry of Science and ICT (MSIT), Republic of Korea (RS-2026-25481217). This research was also supported by the NVIDIA Academic Grant Program using NVIDIA DGX Spark.

\newpage
\bibliographystyle{plainnat}
\bibliography{example_paper}

\clearpage

\appendix
\raggedbottom
\section{Full Proof of Theorem~\ref{thm:unforgeability}}\label{app:full-proof}

We give the step-by-step reduction for Theorem~\ref{thm:unforgeability} (conditions on $\mathcal{A}$ and $q$ are as stated in the main text).

\paragraph{Security game.} Fix the verifier-pinned producer claim $(b,h_{\mathrm{ctx}})$ and let $\mathcal{G}$ contain its honest $(I,\pi)$ pairs. The adversary receives the verification key, public artifacts, and hash-oracle access, but not $k$. The verifier reconstructs the expected public statement, including the canonical image hash, and rejects mismatching proof public values. We consider one candidate image and trace fixed before its first opening-challenge query, as required by Theorem~\ref{thm:unforgeability}; the adversary wins if this candidate is accepted outside $\mathcal{G}$.

\paragraph{Witness extraction.} Under the knowledge-soundness assumption in Assumption~\ref{asm:crypto}, an accepting receipt yields a witness $w^*=(k^*,c^*,\mathcal{T}^*,\mathrm{paths}^*)$ satisfying $\mathcal{P}_1$--$\mathcal{P}_4$, except with error $\mathrm{Adv}^{\mathrm{sound}}_{\mathrm{SP1}}$. The following argument conditions on successful extraction.

\paragraph{Case A: Binding forgery.} The verifier fixes $(b,h_{\mathrm{ctx}})$ independently of the candidate. By $\mathcal{P}_1$, the extracted key satisfies $b=\mathsf{H}(\texttt{ctx\_hash}\|h_{\mathrm{ctx}}\|k^*)$. Recovering such a secret-key witness requires a preimage failure; a distinct key for the pinned input additionally yields a second preimage. These events are included in the aggregate hash-failure term $\mathrm{Adv}^{\mathrm{hash}}_{\mathsf{H}}$.

\paragraph{Case B: Merkle trace forgery.} $\mathcal{P}_2$ requires each leaf $\mathrm{leaf}_i=\mathsf{H}(i,g_i,z_i^{(0)},\mathtt{bit}_i)$ to authenticate to the committed root~\cite{merkle87}. Changing any tuple after fixing that root requires a hash collision, included in $\mathrm{Adv}^{\mathrm{hash}}_{\mathsf{H}}$. A newly committed trace is covered by the partial-opening case below.

\paragraph{Case C: Partial opening with $q$ cheating positions.} Let $F\subseteq S$ contain the $q$ positions violating $\mathcal{P}_4$ in Sec.~\ref{sec:formal-security}: the encoded payload bit, normalized LUT baseline, or Q18 embedding equality is incorrect. Since $\mathcal{P}_4$ checks all three conditions on each opened index, an accepting witness must have $K\cap F=\emptyset$.

The opening set is deterministic in the challenge seed,
\[
\mathrm{challenge\_seed}^* \;=\; \mathsf{H}\!\left(\texttt{ppmark\_opening\_v1} \,\|\, h_{\mathrm{img}}^* \,\|\, h_{\mathrm{ctx}}^* \,\|\, \mathrm{sample\_root}^*\right),
\]
and $K=\mathrm{DeriveOpen}(\mathrm{challenge\_seed}^*,k_{\mathrm{open}},N)$. The candidate's canonical image, trace, and protocol parameters are fixed before this first challenge query. This is an explicit condition of the bound, not a consequence of hashing alone.

Modelling $\mathsf{H}$ as a random oracle~\cite{bellare1993rom} (standard for SHA-256 in this regime), conditionally on all commitments being fixed before the seed is evaluated, $K$ is a uniformly random size-$k_{\mathrm{open}}$ subset of $[N]$. The probability that a uniform $k_{\mathrm{open}}$-subset avoids the $q$-element set $F$ is
\begin{align*}
\Pr[K \cap F = \emptyset] &= \frac{\binom{N-q}{k_{\mathrm{open}}}}{\binom{N}{k_{\mathrm{open}}}}
= \prod_{j=0}^{k_{\mathrm{open}}-1} \frac{N - q - j}{N - j}.
\end{align*}
Since $(N-q-j)/(N-j) \leq (N-q)/N = 1 - q/N$ for every $j \in \{0, \dots, k_{\mathrm{open}}-1\}$ (the ratio is non-increasing in $j$, attaining its maximum $1-q/N$ at $j=0$), the telescoping product yields
\[
\Pr[K \cap F = \emptyset] \;\leq\; (1 - q/N)^{k_{\mathrm{open}}} \;\leq\; \exp\!\left(-k_{\mathrm{open}}\cdot q / N\right).
\]
Conditioned on $K \cap F \neq \emptyset$, $\mathcal{P}_4$ rejects the witness with probability $1$, hence $\pi^*$ fails verification. The contribution of Case~C is therefore $(1 - q/N)^{k_{\mathrm{open}}}$.

\paragraph{Case D: Image binding (cross-image proof transfer).} Consider two images with distinct canonical pixel arrays. If their hashes differ, an existing receipt fails comparison with the target image's expected public statement. If their hashes coincide, this is a hash collision, included in $\mathrm{Adv}^{\mathrm{hash}}_{\mathsf{H}}$. Different seeds need not produce different subsets $K$; receipt binding relies on the checked public statement. Files with identical canonical pixels represent the same artifact for this guarantee.

\paragraph{Union bound.} Combining the failure probabilities from witness extraction and Cases~A--D,
\begin{align*}
\Pr[\mathcal{A}\text{ wins}] \;&\leq\; \mathrm{Adv}^{\mathrm{sound}}_{\mathrm{SP1}}(\mathcal{A}) + \mathrm{Adv}^{\mathrm{hash}}_{\mathsf{H}}(\mathcal{A}) + \frac{\binom{N-q}{k_{\mathrm{open}}}}{\binom{N}{k_{\mathrm{open}}}},
\end{align*}
matching the conditional, per-candidate statement of Theorem~\ref{thm:unforgeability}. \qed

\paragraph{Numerical instantiation.} At $N=1000$, $k_{\mathrm{open}}=32$, and $q=125$, the exact without-replacement miss probability is $0.012965$ (approximately $98.7\%$ detection). The simpler upper bound $(1-q/N)^{32}$ is $0.013940$. This upper bound is approximately $2^{-13.28}$ at $q=250$ and $2^{-32}$ at $q=500$. These are partial-opening probabilities under the theorem's sampling conditions; cryptographic failure terms remain separate.

\paragraph{Small-$q$ regime ($q < 0.125 N$).} The upper bound $(1-q/N)^{32}$ ranges from approximately $0.014$ to $0.968$ for $q\in\{1,\ldots,125\}$, so partial opening alone gives a weak guarantee for small $q$. Score checks all $N$ positions, but clean-FPR calibration does not establish a general adversarial FAR bound. The submitted per-step transfer FARs remain $0$--$0.01$ (Sec.~\ref{sec:imprint-forgery}); the separate white-box FAR results remain as reported in Table~\ref{tab:white-box-summary}. These empirical results do not prove a small-$q$ security bound.

\paragraph{Opening-tolerance parameter $\varepsilon$.} We use $\varepsilon=0$: $\mathcal{P}_4$ independently reconstructs the normalized Gaussian baseline and payload sign, then checks exact equality under the declared Q18 arithmetic. A nonzero tolerance changes the verified relation and would require a separate security analysis.

\section{Verification Design Details}\label{app:verification-design}

This appendix complements Sec.~\ref{sec:synchronization} with the design rationale for using a cryptographic hash rather than a perceptual hash for $h_{\mathrm{img}}$, and Sec.~\ref{sec:sig-vs-zkp} with image-hash signature controls.

\paragraph{Cryptographic hash vs perceptual hash for proof binding.}
We use SHA-256 over canonical pixels for $h_{\mathrm{img}}$ rather than a perceptual hash. The verifier compares the receipt's public statement with the expected image hash, so a change to canonical pixels invalidates an existing receipt except for a hash collision. A perceptual identifier shared by distinct images would weaken this exact-image binding.

Benign edits that change canonical pixels therefore invalidate the original Accept receipt. Routine Score screening of a derivative requires no re-proof; a fresh receipt is needed only to grant Accept to that derivative itself. 17.0\,s projected receipt issuance at $N=1000$ (Table~\ref{tab:proof-backend-comparison}). Exact integrity concerns canonical pixels rather than the original file-container bytes.

\subsection{Image-hash signature controls}\label{app:signature-controls}

The Sig-only and Sig+Score baselines in Sec.~\ref{sec:sig-vs-zkp} authenticate metadata. We additionally evaluate two domain-separated Ed25519 controls with a fixed provider public key and the same canonical-pixel hash $h_{\mathrm{img}}$: \textbf{Bare-Sig} signs $h_{\mathrm{img}}$; \textbf{Context-Sig} signs $h_{\mathrm{img}} \parallel h_{\mathrm{ctx}} \parallel b$.

\begin{table}[H]
\centering
\caption{Image-hash signature controls ($n{=}100$ per test).}
\label{tab:signature-controls}
{\ppmarksingletablesetup
\begin{tabular}{@{}>{\raggedright\arraybackslash}p{0.43\linewidth}cc>{\raggedright\arraybackslash}p{0.24\linewidth}@{}}
\toprule
Test & Bare-Sig & Context-Sig & PP-Mark evidence \\
\midrule
Original verifies & 100/100 & 100/100 & --- \\
Cross-image replay & 0/100 & 0/100 & --- \\
Context substitution & 100/100 & 0/100 & --- \\
\addlinespace[2pt]
Authorized signer signs an unwatermarked image & 100/100 & 100/100 & 0/100 (Score/Accept) \\
\addlinespace[2pt]
JPEG ($q{=}25$) after signing & 0/100 & 0/100 & 94/100 (Score) \\
75\% crop--resize after signing & 0/100 & 0/100 & 87/100 (Score) \\
\bottomrule
\end{tabular}
\par}
\vspace{0.3em}
{\small\raggedright
Signature columns count valid signatures. PP-Mark entries are the separately reported compliance-bypass results (Sec.~\ref{sec:sig-vs-zkp}) and Score detections (Table~\ref{tab:image-transforms}).
\par}
\end{table}

For exact canonical-image integrity under a trusted signer, an image-hash signature has the same modification-rejection role as Accept's image-hash binding and is simpler and preferable. Context-Sig also prevents context substitution, but an authorized signer can still sign an image without embedding a watermark. A signature authenticates the signer's assertion; it does not verify that embedding occurred. Accept additionally requires both a passing watermark score and a valid receipt enforcing $\mathcal{P}_1$--$\mathcal{P}_4$, including sampled embedding consistency. Score provides transform-tolerant in-image detection.

\section{Empirical Verification of Assumption~\ref{assum:noise}}\label{app:assumption-verify}

We verify the two conditions of Assumption~\ref{assum:noise} on $n=100$ watermarked images by computing the DDIM inversion residual $\epsilon = \hat{z}_0 - z_0$ (channel~0) for each image, where $z_0$ is the watermarked latent saved during generation and $\hat{z}_0$ is the latent recovered via DDIM inversion from the decoded pixel image. We report results on both SD\,2.1 (latent $64{\times}64$) and SDXL (latent $128{\times}128$).

\paragraph{Zero-mean condition.}
On SD\,2.1 the per-image residual mean is $\bar{\epsilon} = +0.011 \pm 0.003$ (grand mean $\pm$ SE across 100 images), two orders of magnitude below the residual standard deviation $\sigma_\epsilon \approx 0.50$, so $|\bar{\epsilon}|/\sigma_\epsilon \approx 0.02$.
On SDXL the grand mean is $\bar{\epsilon} = -0.022 \pm 0.004$ (95\%~CI $[-0.031,-0.013]$) with $\sigma_\epsilon \approx 0.44$, yielding $|\bar{\epsilon}|/\sigma_\epsilon \approx 0.05$.
Each channel-0 residual contains 4,096 entries on SD\,2.1 and 16,384 on SDXL. The bias is statistically detectable across the 100 images but small relative to both the residual standard deviation and $\alpha_{\mathrm{eff}}$; its magnitude alone does not establish a zero effect on the score.

\paragraph{Bounded correlation condition.}
On SD\,2.1 the mean Pearson correlation is $|\rho(\epsilon, s)| = 0.30$ (median~$0.30$, max~$0.48$).
On SDXL the correlation is larger: $|\rho| = 0.54$ (median~$0.55$, max~$0.66$). This two-model comparison does not isolate latent dimensionality as the cause of the difference.
Per-image correlations are negative in every case in both models, so the attenuation acts as a multiplicative shrinkage on the score rather than as random noise.
The per-image $p$-values are small (all $<\!0.05$ on SDXL). We therefore report correlation magnitudes and observed score margins; correlation below one alone is not a sufficient detection guarantee.
On SD\,2.1 the watermark signal contribution $\alpha_{\mathrm{eff}}\sqrt{N}=30.7$ is an expectation reference, not a hard score maximum. The observed mean is $25.7$ (approximately $16\%$ lower), and the minimum score is $17.4$, about $7.0\times$ the threshold $\tau=2.47$ across the 100 images.
On SDXL the observed mean clean score is $13.6$ ($\mathrm{score}/\tau \approx 5.5$), again with every clean image crossing $\tau$ at FPR${\leq}$1\%.
These measurements support the residual approximation on the evaluated images; both models retain 100/100 detections despite their different score margins.

Figure~\ref{fig:assumption-verify} overlays the SD\,2.1 and SDXL distributions on shared axes: (a)~per-image residual means are centered near zero for both models ($\bar{\epsilon}=+0.011$ on SD\,2.1, $-0.022$ on SDXL), and (b)~per-image $|\rho(\epsilon, s)|$ concentrates around $0.30$ for SD\,2.1 and shifts to approximately $0.54$ for SDXL while remaining bounded well below~$1$.

\begin{figure}[!htbp]
\centering
\includegraphics[width=0.8\linewidth]{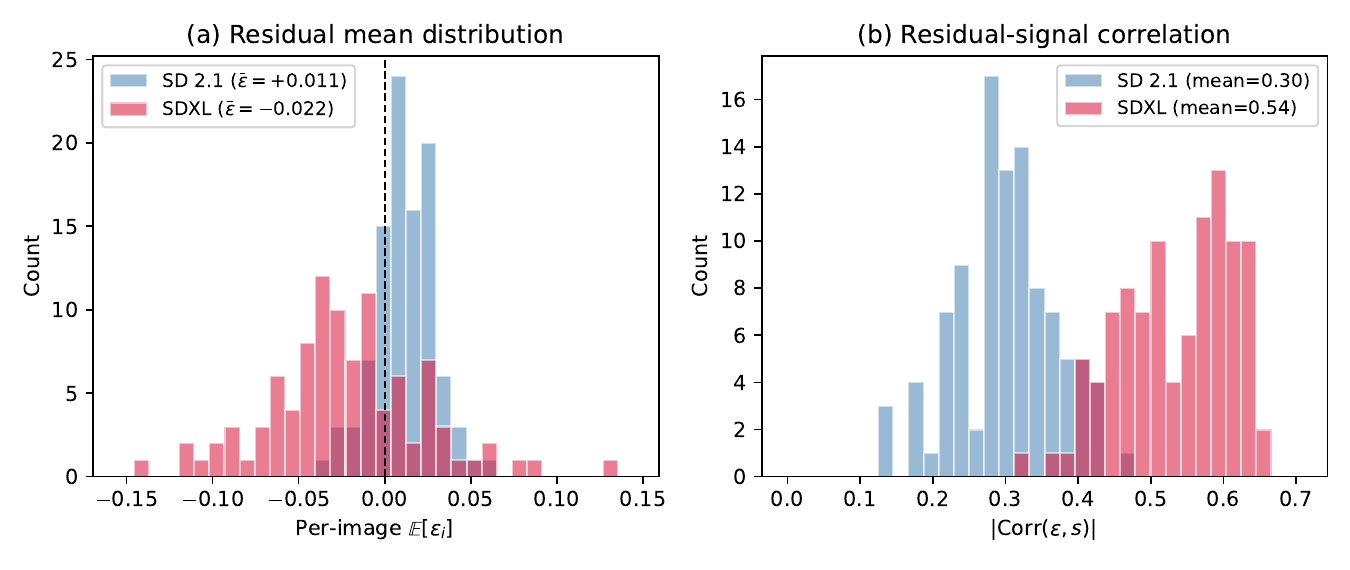}
\caption{Empirical assessment of Assumption~\ref{assum:noise} on SD\,2.1 and SDXL ($n{=}100$ each, overlaid). (a)~Per-image residual means cluster near zero. (b)~Mean $|\mathrm{Corr}(\epsilon,s)|$ is $0.30$ on SD\,2.1 and $0.54$ on SDXL; the two-model comparison does not isolate the cause of this difference.}
\label{fig:assumption-verify}
\end{figure}

\section{Ablation Study}
\subsection{Sample Size and Statistical Reliability}\label{app:sample-size}
\begin{figure}[t]
    \centering
    \begin{subfigure}{0.48\linewidth}
        \centering
        \includegraphics[width=\linewidth]{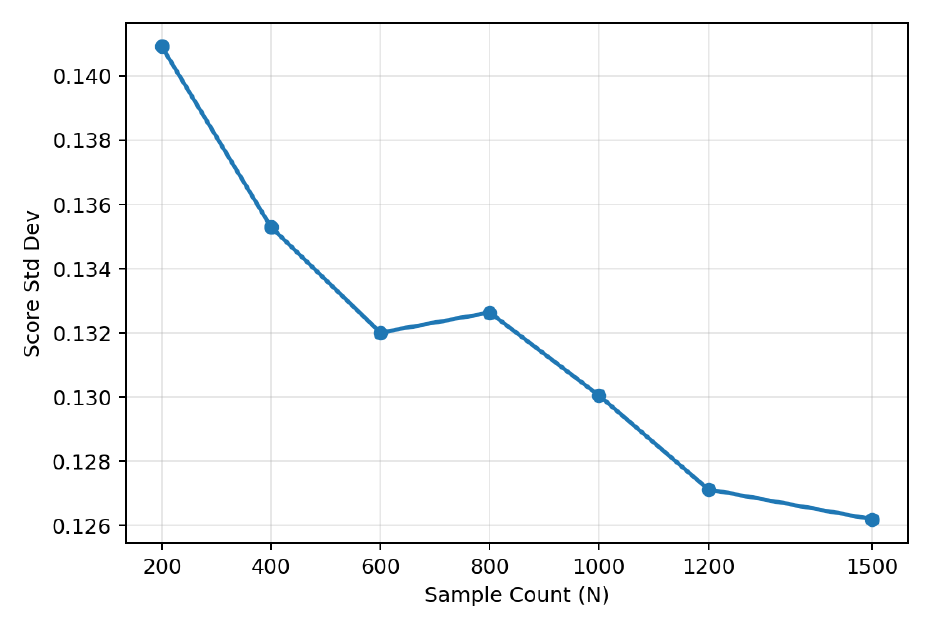}
        \caption{Score concentration: standard deviation vs.\ sample count $N$.}
        \label{fig:score-std-n}
    \end{subfigure}
    \hfill
    \begin{subfigure}{0.48\linewidth}
        \centering
        \includegraphics[width=\linewidth]{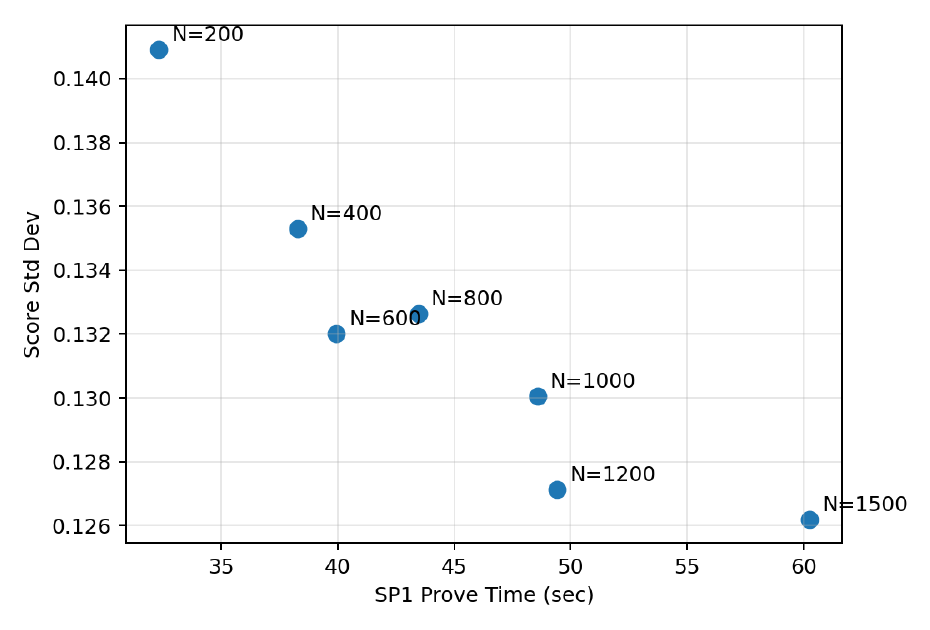}
        \caption{Knee trade-off between proving cost and score variability.}
        \label{fig:knee-tradeoff}
    \end{subfigure}
    \caption{Sample size ablation. (a) Score variance decreases as $N$ grows; (b) knee point identifies $N{=}1000$ as the trade-off between proving cost and score stability.}
    \label{fig:sample-size-ablation}
\end{figure}

We analyze the trade-off between statistical reliability and computational overhead by varying the number of sampled latent positions $N \in \{200, \dots, 1500\}$. We identify $N=1000$ as the reasonable point where statistical stability is maximized without incurring excessive prover latency or bandwidth costs.

\textbf{Statistical Concentration:} Theoretically, the detection score (cosine similarity) acts as a smooth function of bounded aggregates, suggesting Hoeffding-sense concentration. Empirically, the score's standard deviation across images decreases by 7.7\% as $N$ increases from 200 to 1000 ($0.1409 \rightarrow 0.1301$). However, increasing $N$ further to 1500 yields only a marginal 3.0\% additional reduction, indicating diminishing returns in variance reduction.
  
\begin{table}[H]
\centering
\caption{SP1 timing vs.\ sample count $N$. Submitted v5 measurements are retained; the final row shows the latest configuration with explicitly marked time estimates.}
\label{tab:sp1-k-sweep-timing}
{\ppmarksingletablesetup
  \begin{tabular}{c S[table-format=2.3, round-precision=3] S[table-format=1.3, round-precision=3]}
    \toprule
    \textbf{N} & {\textbf{Prove / issue (s)}} & {\textbf{Verify (s)}} \\
    \midrule
    200  & 32.309 & 4.468 \\
    400  & 38.278 & 5.640 \\
    600  & 39.932 & 5.886 \\
    800  & 43.465 & 6.651 \\
    1000 & 48.572 & 7.589 \\
    1200 & 49.425 & 8.877 \\
    1500 & 60.272 & 9.802 \\
    \midrule
    $1000^{\dagger}$ & \multicolumn{1}{c}{$17.0^{\dagger}$} & \multicolumn{1}{c}{$4.0^{\dagger}$} \\
    \bottomrule
  \end{tabular}
\par}
\vspace{2pt}
{\small\raggedright $^{\dagger}$SP1 v6.3.1, latest circuit implementation. The 17.0\,s, 4.0\,s values are a projected fresh-process receipt-issuance phase total, including setup, execution, proving, self-verification, and serialization.\par
\par}
\end{table}

\begin{figure}[H]
\centering
\includegraphics[width=0.6\linewidth, trim=0 0 0 25pt, clip]{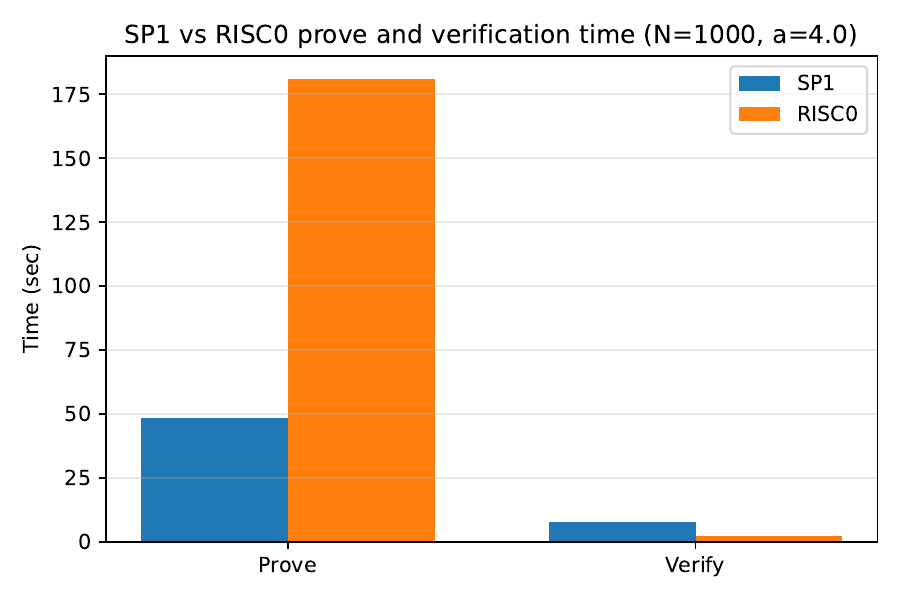}
\caption{Submitted prover and verifier timings at $N{=}1000$ ($\alpha{=}4.0$). SP1 uses a single H100 GPU for proving with CPU receipt verification; RISC0 uses the same SD2.1 setup and public inputs with CPU receipt verification.}
\label{fig:sp1-vs-risc0-timing}
\end{figure}

\textbf{Computational Cost:} In the submitted SP1 v5 configuration, prover and verification times increase with $N$ (Table~\ref{tab:sp1-k-sweep-timing}): proving rises from 32.3\,s ($N=200$) to 48.6\,s ($N=1000$). Figure~\ref{fig:sp1-vs-risc0-timing} retains the submitted SP1/RISC0 comparison. The appended v6.3.1 row uses a different implementation and timing boundary and is not part of the original $N$ sweep.

\subsection{Runtime Comparison}\label{app:runtime}
Table~\ref{tab:runtime} retains the submitted SD2.1 embedding and detection timings and updates the Accept row to the latest rebuttal configuration. Semantic methods embed during diffusion generation; post-hoc methods apply a separate encoder. Updated receipt issuance is approximately 17.0\,s projected at $N=1000$. Table~\ref{tab:sp1-k-sweep-timing} gives a conditional estimate of standalone receipt verification under unchanged fixed overhead. Total Accept-verification time has not been jointly measured and is left unreported.

\begin{table}[H]
\centering
\caption{Runtime comparison on SD2.1. Submitted method timings are retained; the Accept row uses the latest rebuttal receipt-issuance value.}
\label{tab:runtime}
{\small
\setlength{\tabcolsep}{6pt}%
\renewcommand{\arraystretch}{1.05}%
\begin{tabular}{lcc}
\toprule
Method & Embed (s) & Verify (s) \\
\midrule
Tree-Ring & 0.89 & 0.75 \\
Gaussian Shading & 3.10 & 0.72 \\
RingID & 0.89 & 0.72 \\
WIND & 0.89 & 1.10 \\
Stable Signature & 0.91 & 0.008 \\
HiDDeN & 0.001 & 0.001 \\
TrustMark & 0.013 & 0.004 \\
InvisMark & 0.002 & 0.009 \\
\midrule
PP-Mark (score) & 0.89 & 0.72 \\
PP-Mark (accept, v6.3.1) & 0.89\,{+}\,17.0$^{\dagger}$ & 0.72\,{+}\,4.0$^{\dagger}$ \\
\bottomrule
\end{tabular}
}
\vspace{2pt}

{\small\raggedright $^{\dagger}$The 17.0\,s 4.0\,s  are projected receipt issuance, including initialization and self-verification, rather than standalone proving. The displayed sum combines the submitted embedding time with that projection; it is not a jointly timed run. Total Accept verification is not jointly measured; the separate receipt-verification estimate is in Table~\ref{tab:sp1-k-sweep-timing}.\par}
\end{table}

\subsection{SP1 Receipt Size}\label{app:proof-size}
Table~\ref{tab:sp1-receipt-size} retains the submitted size sweep and appends the latest SP1 v6.3.1 Core artifact (2.81\,MB at $N=1000$). Table~\ref{tab:proof-backend-comparison} reproduces the final rebuttal's backend comparison. The post-submission implementation uses an optimized full-trace streaming commitment; the original Merkle configuration remains the basis of the submitted sweep. Version, implementation, and timing-boundary differences preclude attributing the change to the SP1 release alone.

\begin{table}[H]
\centering
\caption{SP1 receipt byte sizes across $N$ (SD\,2.1, $\alpha{=}4.0$). Submitted rows are retained, followed by the latest rebuttal configuration.}
\label{tab:sp1-receipt-size}
{\ppmarksingletablesetup
  \begin{tabular}{c S[table-format=4.0, round-precision=0] S[table-format=4.2, round-precision=2] S[table-format=2.2, round-precision=2]}
    \toprule
    \textbf{N} & {\makecell{\textbf{Public}\\ \textbf{inputs (B)}}} & {\makecell{\textbf{Witness}\\ \textbf{(KB)}}} & {\makecell{\textbf{Receipt}\\ \textbf{(MB)}}} \\
    \midrule
    300  & 674 &  254.69 & 12.24 \\
    400  & 664 &  339.21 & 16.87 \\
    600  & 665 &  553.15 & 24.81 \\
    1000 & 666 &  921.60 & 40.68 \\
    1500 & 666 & 1493.57 & 65.96 \\
    \midrule
    $1000^{\dagger}$ & \multicolumn{1}{c}{1,004} & \multicolumn{1}{c}{117.0} & 2.81 \\
    \bottomrule
  \end{tabular}
\par}
\vspace{2pt}
{\small\raggedright $^{\dagger}$Latest SP1 v6.3.1 Core implementation, updated public-input, witness, receipt sizes.\par}
\end{table}

\begin{table}[H]
\centering
\caption{Proof-backend comparison from the final rebuttal ($N=1000$), including the latest circuit implementation.}
\label{tab:proof-backend-comparison}
{\ppmarksingletablesetup
\begin{tabular}{lrr}
\toprule
Proof configuration & Time (s) & Artifact \\
\midrule
SP1 v6.3.1 $^{\dagger}$ & $17.0^{\dagger}$ & 2.81 MB \\
SP1 v5 & 48.572 & 40.68 MB \\
RISC0 & 180.9 & 40.51 MB \\
PLONK & 312.677 & 3,874 B \\
Groth16 & 438.380 & 1,499 B \\
\bottomrule
\end{tabular}
\par}
\vspace{2pt}
{\small\raggedright $^{\dagger}$SP1 v6.3.1 uses a Core receipt. Its 17.0\,s value is the projected fresh-process issuance phase total; the 2.81\,MB artifact size is measured. The other rows retain the times and sizes reported.\par}
\end{table}

\subsection{Watermark Strength}
We study the effect of watermark strength $\alpha$ on visual quality while fixing the sampling budget at $N=1000$. We report BRISQUE (lower is better) and focus on the 95th percentile across images to capture tail behavior. For each $\alpha \in \{3.0, 3.5, 4.0, 4.5\}$ we evaluate 100 watermarked images under identical settings. 

The BRISQUE p95 is relatively stable from $\alpha=3.0$ to $4.0$ (roughly $34.1$--$34.9$), but it jumps at $\alpha=4.5$ to $40.5$, indicating a heavier quality tail. This suggests that pushing $\alpha$ beyond 4.0 yields a disproportionate increase in worst-case artifacts without a clear compensating gain in reliability, given that $N=1000$ already stabilizes detection. We therefore select $\alpha=4.0$ as a conservative balance between robustness and perceptual quality (Fig.~\ref{fig:alpha-brisque-p95}).

\begin{figure}[H]
  \centering
  \includegraphics[width=0.6\linewidth]{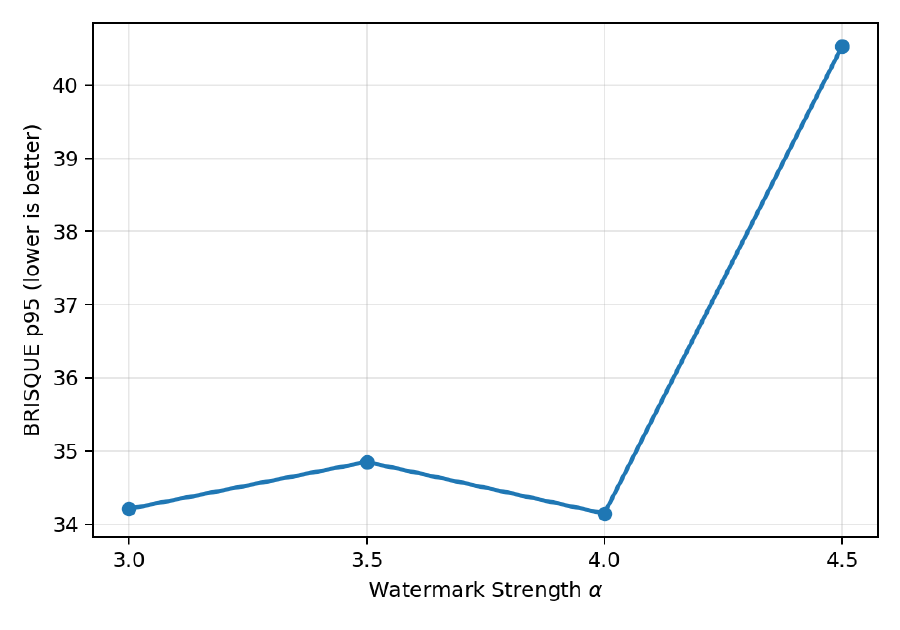}
  \caption{BRISQUE p95 vs.\ watermark strength $\alpha$ (lower is better).}
  \label{fig:alpha-brisque-p95}
\end{figure}

\section{Computational Resources}\label{app:compute}
The submitted experiments were conducted on an HPE server equipped with a single NVIDIA H100 PCIe GPU (80\,GB) and an Intel\textsuperscript{\textregistered} Xeon\textsuperscript{\textregistered} Gold 6542Y CPU (24 cores, 48 threads, 2.9\,GHz, 60\,MB cache). The GPU handles image generation, watermark embedding, latent inversion, attack optimization (PGD, SPSA, imprint forgery, regeneration), and SP1 STARK proving; SP1 receipt verification runs on the CPU. Per-method runtimes are in Appendix~\ref{app:runtime}, and the submitted $N$-dependent prove/verify sweep is in Appendix~\ref{app:sample-size} (Table~\ref{tab:sp1-k-sweep-timing}). Post-submission configurations are distinguished in Table~\ref{tab:proof-backend-comparison}.

\section{Imprint Forgery Details}\label{app:imprint-forgery-details}

This appendix complements Sec.~\ref{sec:imprint-forgery} with the full step-wise breakdowns for both the Direct and Transfer imprint-forgery experiments, along with the bit-accuracy calibration used for binary-only detectors (TrustMark, InvisMark). The aggregated cross-method comparison (FAR with 95\% confidence intervals, $\Delta\%$ ratios, and one-sided Fisher exact tests) and the corresponding forest plot remain in the main paper (Table~\ref{tab:transfer-summary}, Figure~\ref{fig:transfer-forest}).

\subsection{Direct Imprint Forgery}\label{app:direct-imprint}

We sanity-check the standard imprint-forgery optimizer on Tree-Ring (TR) and Gaussian Shading (GS) themselves to confirm the attack pipeline functions correctly before applying it cross-method. For each fixed cover--reference pair, we run the optimizer for $T \in \{50, 100, 150\}$ steps and report the Forgery Acceptance Rate at 1\% FPR (FAR@1\%FPR, $n{=}100$ per step) using each method's native detector with thresholds calibrated to 1\% FPR on clean (non-watermarked) images.

\sisetup{
  detect-weight=true,
  table-number-alignment=center,
  round-mode=places,
  round-precision=2
}

\begin{table}[H]
\centering
\caption{Direct vulnerability under targeted imprint forgery. FAR@1\%FPR (targeted verify); all methods evaluated at FPR=1\% (clean-calibrated).}
\label{tab:direct-vulnerability}
{\small
  \setlength{\tabcolsep}{6pt}%
  \renewcommand{\arraystretch}{1.05}%
  \begin{tabular}{
  c
  S[table-format=1.3, round-mode=none] S[table-format=1.3, round-mode=none]
  }
    \toprule
    &
    \multicolumn{1}{c}{\textbf{TR}} &
    \multicolumn{1}{c}{\textbf{GS}} \\
    \cmidrule(lr){2-2}\cmidrule(lr){3-3}
    \textbf{Step}
    & {\textbf{FAR@1\%FPR}} & {\textbf{FAR@1\%FPR}} \\
    \midrule
    50 & 0.92 & 1.00 \\
    100 & 0.97 & 1.00 \\
    150 & 0.97 & 1.00 \\
    Avg. & 0.953 & 1.00 \\
    \bottomrule
  \end{tabular}
}
\end{table}

Both methods are highly vulnerable across all step budgets. Tree-Ring's FAR rises from 0.92 (50 steps) to 0.97 (100--150 steps), averaging 0.953 across the three step counts; Gaussian Shading saturates at 1.00 across all step counts, indicating near-perfect imprint transfer onto arbitrary covers. These rates confirm the optimizer reliably extracts and re-injects watermark imprints, providing a sound baseline for cross-method evaluation in App.~\ref{app:transfer-imprint}.

\subsection{Transfer Imprint Forgery}\label{app:transfer-imprint}

We extend the imprint-forgery pipeline cross-method: imprints extracted from a TR- or GS-watermarked reference are optimized onto each fixed cover and evaluated against each target method's native detector. All thresholds are calibrated to 1\% FPR on clean images, so the reported differences reflect detector susceptibility under identical pairings, step budgets, and calibration. Step budget is $T \in \{50, 100, 150\}$ with $n{=}100$ pairs per step.

\paragraph{Bit-accuracy calibration for binary-only detectors.}
TrustMark \cite{bui2023trustmark} and InvisMark \cite{xu2024invismark} expose only binary detection flags through their public APIs; standard threshold-based FPR calibration is therefore not directly available. We instead use bit accuracy, the fraction of correctly decoded payload bits, as a continuous score, enabling threshold-based FPR=1\% calibration on clean images. The resulting calibration thresholds are $\tau{=}0.58$ for TrustMark (clean bit accuracy ${\approx}0.48$, watermarked ${\approx}1.00$) and $\tau{=}0.49$ for InvisMark (clean ${\approx}0.41$, watermarked ${\approx}0.96$). The large clean--watermarked gaps confirm bit accuracy provides a reliable continuous score for these methods despite their binary-only public interfaces.

\begin{table}[H]
  \centering
  \caption{Transfer results by method and step. FAR@1\%FPR (targeted verify); all methods evaluated at FPR=1\% (clean-calibrated). TrustMark and InvisMark use bit accuracy as the continuous score.}
  \label{tab:transfer-steps}
  {\footnotesize
    \setlength{\tabcolsep}{2.0pt}%
    \renewcommand{\arraystretch}{1.04}%
    \begin{tabular*}{\linewidth}{@{\extracolsep{\fill}}
    c c
    S[table-format=1.2, round-precision=2] S[table-format=1.2, round-precision=2] S[table-format=1.2, round-precision=2] S[table-format=1.2, round-precision=2]
    S[table-format=1.2, round-precision=2] S[table-format=1.2, round-precision=2] S[table-format=1.2, round-precision=2] S[table-format=1.2, round-precision=2]
    @{}}
      \toprule
      \textbf{Source} & \textbf{Step}
      & {\textbf{TrustMark}} & {\textbf{InvisMark}} & {\textbf{RingID}} & {\textbf{StableSig}}
      & {\textbf{HiDDeN}} & {\textbf{WIND}} & {\textbf{PP-Mark$_{\mathrm{score}}$}} & {\textbf{PP-Mark$_{\mathrm{accept}}$}} \\
      \midrule

      \multirow{3}{*}{TR} & 50 & 0.05 & 0.06 & 0.01 & 0.04 & 0.04 & 0.01 & 0.00 & 0.00 \\
                          & 100 & 0.03 & 0.01 & 0.10 & 0.04 & 0.02 & 0.01 & 0.00 & 0.00 \\
                          & 150 & 0.03 & 0.02 & 0.01 & 0.02 & 0.01 & 0.01 & 0.01 & 0.00 \\
      \addlinespace
      \multirow{3}{*}{GS} & 50 & 0.01 & 0.03 & 0.01 & 0.01 & 0.01 & 0.01 & 0.00 & 0.00 \\
                          & 100 & 0.02 & 0.01 & 0.02 & 0.02 & 0.03 & 0.01 & 0.01 & 0.00 \\
                          & 150 & 0.05 & 0.01 & 0.01 & 0.02 & 0.03 & 0.01 & 0.01 & 0.00 \\

      \bottomrule
    \end{tabular*}
  }
\end{table}

Across all step budgets and both imprint sources (TR, GS), PP-Mark$_{\mathrm{score}}$ remains at 0.00--0.01, indicating that the standard imprint-forgery pipeline rarely pushes its detection score beyond the 1\% calibration floor. The peak baseline value is RingID at 0.10 under TR-sourced imprints with 100 optimization steps; TrustMark, InvisMark, StableSig, and HiDDeN range across 0.01--0.06 across the various source--step combinations. Among baselines, WIND remains the closest competitor to PP-Mark$_{\mathrm{score}}$ with FAR consistently at 0.01, while PP-Mark$_{\mathrm{accept}}$ stays at 0.00 across all settings, confirming that the cryptographic verification rejects any imprint-altered image regardless of score-domain optimization.

\section{Evaluation on SDXL}\label{app:sdxl}

We further evaluate PP-Mark on Stable Diffusion~XL (SDXL~1.0), a second latent-diffusion model sharing the initial-latent embedding and DDIM-inversion interface with Stable Diffusion~2.1 (SD2.1). SDXL generates $1024{\times}1024$ images with a $128{\times}128$ latent space (four times the dimensionality of SD2.1's $64{\times}64$ latent). Among the baselines from the main experiments, only TrustMark and InvisMark are natively compatible with SDXL's resolution; the remaining methods (RingID, StableSig, HiDDeN, WIND) either predate SDXL or require resolution downsampling that renders their outputs non-comparable. We therefore scope this evaluation to PP-Mark, TrustMark, and InvisMark.

\paragraph{Setup.}
We generate 100 clean and 100 watermarked SDXL images per method (50 DDIM steps). PP-Mark uses $\alpha{=}4.0$, sample count $N{=}1000$ with the 128${\times}$128 latent. All thresholds are calibrated at FPR${\leq}$1\% on SDXL clean images ($n{=}100$).

\paragraph{Detection.}
All three methods achieve TPR$=$1.00 on their respective SDXL watermarked images ($n{=}100$ each). PP-Mark's watermarked scores (mean$=$13.63) exceed $\tau{=}2.47$ by $5.5{\times}$, confirming that the larger latent space preserves detection margin.

\paragraph{Direct imprint vulnerability (TR/GS sanity check).}
As in SD2.1 (Table~\ref{tab:direct-vulnerability}), we first verify that the SDXL imprint-forgery pipeline reaches each semantic watermark's own detector. Table~\ref{tab:sdxl-direct-vulnerability} reports FAR@1\%FPR for Tree-Ring and Gaussian Shading under direct attack ($n{=}100$ per step). Both are highly vulnerable (TR: $0.83$--$0.91$, GS: $1.00$ across steps), confirming the SDXL attack pipeline is functional before evaluating pixel-space and post-hoc methods below.

\begin{table}[!htbp]
    \caption{Direct vulnerability under targeted imprint forgery on SDXL (sanity check).
    FAR@1\%FPR (targeted verify); TR and GS evaluated at their own clean-calibrated thresholds.}
    \label{tab:sdxl-direct-vulnerability}
    \begin{center}
      \begin{small}
        \begin{sc}
          \begin{tabular}{
          c
          S[table-format=1.3, round-mode=none] S[table-format=1.3, round-mode=none]
          }
            \toprule
            &
            \multicolumn{1}{c}{\textbf{TR}} &
            \multicolumn{1}{c}{\textbf{GS}} \\
            \cmidrule(lr){2-2}\cmidrule(lr){3-3}
            \textbf{Step}
            & {\textbf{FAR@1\%FPR}} & {\textbf{FAR@1\%FPR}} \\
            \midrule
            50  & 0.83 & 1.00 \\
            100 & 0.90 & 1.00 \\
            150 & 0.91 & 1.00 \\
            Avg. & 0.880 & 1.00 \\
            \bottomrule
          \end{tabular}
        \end{sc}
      \end{small}
    \end{center}
    \vskip -0.1in
  \end{table}

\paragraph{Transfer imprint forgery.}
We apply the same imprint-forgery protocol as Section~\ref{sec:imprint-forgery} using Tree-Ring (TR) and Gaussian Shading (GS) as imprint sources ($n{=}100$ per step). Table~\ref{tab:sdxl-transfer} reports FAR at the calibrated thresholds.

PP-Mark$_{\mathrm{score}}$ achieves FAR$=$0.00--0.02, consistent with SD2.1 (0.00--0.01). PP-Mark$_{\mathrm{accept}}$ remains at 0.00 across all conditions, as the ZKP proof verification deterministically rejects forged images. TrustMark remains robust (FAR$\leq$0.05). InvisMark shows materially higher vulnerability on SDXL (FAR$=$0.09--0.17) compared to SD2.1 (0.01--0.06), suggesting that SDXL's larger latent space amplifies imprint transfer susceptibility for pixel-space decoders. Figure~\ref{fig:sdxl-transfer-forest} visualizes pooled FAR with Clopper--Pearson 95\% confidence intervals: PP-Mark$_{\mathrm{score}}$ 0.01 [0.00, 0.03], TrustMark 0.03 [0.01, 0.05], InvisMark 0.13 [0.09, 0.17] for TR (GS similar).

\begin{table}[H]
    \caption{Transfer imprint forgery on SDXL.
    FAR at FPR${\leq}$1\% thresholds calibrated on SDXL clean ($n{=}100$ per step).}
    \label{tab:sdxl-transfer}
    \centering
      {\ppmarksingletablesetup
          \begin{tabular}{
          l c
          S[table-format=1.2, round-precision=2] S[table-format=1.2, round-precision=2]
          }
            \toprule
            & &
            \multicolumn{1}{c}{\textbf{TR}} &
            \multicolumn{1}{c}{\textbf{GS}} \\
            \cmidrule(lr){3-3}\cmidrule(lr){4-4}
            \textbf{Method} & \textbf{Step}
            & {\textbf{FAR@1\%FPR}} & {\textbf{FAR@1\%FPR}} \\
            \midrule

            \multirow{3}{*}{TrustMark} & 50  & 0.03 & 0.00 \\
                                       & 100 & 0.05 & 0.02 \\
                                       & 150 & 0.00 & 0.01 \\
              \addlinespace

            \multirow{3}{*}{InvisMark} & 50  & 0.16 & 0.15 \\
                                       & 100 & 0.09 & 0.11 \\
                                       & 150 & 0.14 & 0.17 \\
              \addlinespace

            \multirow{3}{*}{PP-Mark$_{\mathrm{score}}$} & 50  & 0.02 & 0.01 \\
                                       & 100 & 0.00 & 0.01 \\
                                       & 150 & 0.01 & 0.01 \\
              \addlinespace

            \multirow{3}{*}{PP-Mark$_{\mathrm{accept}}$} & 50  & 0.00 & 0.00 \\
                                       & 100 & 0.00 & 0.00 \\
                                       & 150 & 0.00 & 0.00 \\

            \bottomrule
          \end{tabular}
      \par}
  \end{table}

\begin{figure}[!htbp]
    \centering
    \includegraphics[width=0.8\linewidth]{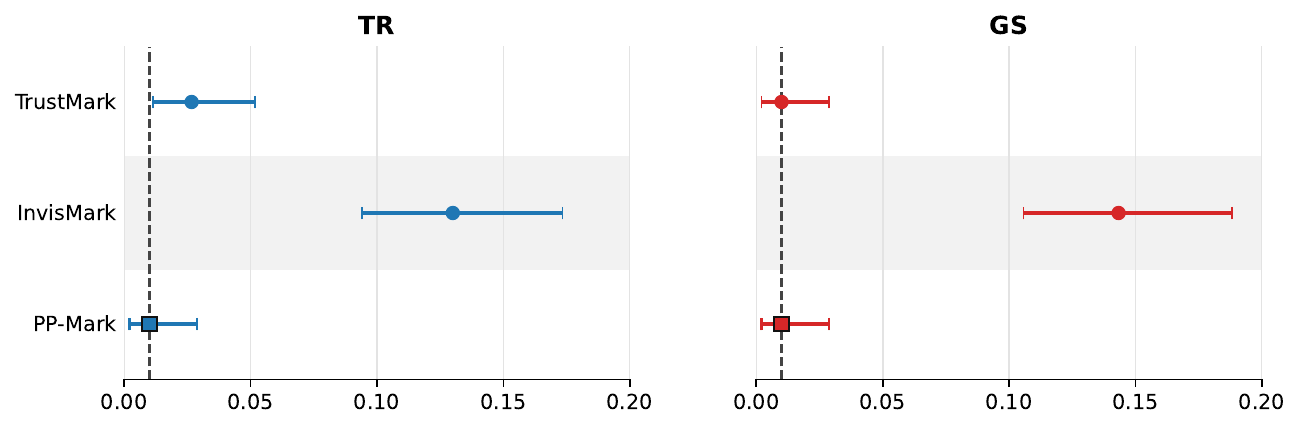}
    \caption{SDXL transfer imprint forgery: pooled FAR across 3 steps ($n{=}300$ per source). Error bars denote Clopper--Pearson 95\% CI; the dashed line marks PP-Mark$_{\mathrm{score}}$.}
    \label{fig:sdxl-transfer-forest}
\end{figure}

\begin{table}[H]
  \caption{White-box PGD forgery on SDXL ($\epsilon{=}2/255$, $T{=}150$, $n{=}100$). FAR at FPR${\leq}$1\% under any-step acceptance; Avg.\ PSNR averaged against the original cover image at the submitted-image step (first-breach step for acceptances, final step otherwise).}
  \label{tab:sdxl-wb}
  \centering
  {\ppmarksingletablesetup
        \begin{tabular}{l
          S[table-format=1.2, round-precision=2]
          S[table-format=2.2, round-precision=2]
          S[table-format=2.2, round-precision=2]
        }
          \toprule
          \textbf{Method}
          & {\makecell{\textbf{FAR}}}
          & {\makecell{\textbf{Avg.\ Steps}\\ \textbf{to FA}}}
          & {\makecell{\textbf{Avg.\ }\\ \textbf{PSNR $\uparrow$}}} \\
          \midrule
          TrustMark                          & 1.00 & 2.78  & 80.28 \\
          InvisMark                          & 1.00 & 3.83  & 71.22 \\
          PP-Mark$_{\mathrm{score}}$         & 0.11 & 29.55 & 85.63 \\
          PP-Mark$_{\mathrm{accept}}$        & 0.00 & \multicolumn{1}{c}{--} & \multicolumn{1}{c}{--} \\
          \bottomrule
        \end{tabular}
  \par}
\end{table}

\paragraph{White-box PGD forgery.}
We apply the same PGD attack ($\epsilon{=}2/255$, $T{=}150$) to all three methods on SDXL, using the any-step acceptance criterion (attack succeeds if $\mathrm{score}{\geq}\tau$ at any step) as in SD2.1. Table~\ref{tab:sdxl-wb} reports outcomes and Figure~\ref{fig:wb_forge_far_curves_sdxl} plots the cumulative FAR trajectories.

TrustMark and InvisMark are breached at 100\% within 3--4 steps on average, consistent with SD2.1. PP-Mark$_{\mathrm{score}}$ achieves FAR$=$0.11 ($n{=}100$, 95\% Clopper--Pearson CI $[5.62\%, 18.83\%]$), a substantial improvement over SD2.1 (FAR$=$0.76). The 4${\times}$ larger latent dimensionality (65,536 vs 16,384) may contribute to this difference, but the two-model comparison does not isolate its effect. The models also differ in PSNR: PP-Mark$_{\mathrm{score}}$'s average rises from 67.69 (SD2.1) to 85.63 (SDXL), whereas TrustMark/InvisMark PSNR moves only slightly (+4.5/$-$1.5 dB). Although the attacker fills the full $\ell_\infty$ budget ($\epsilon{=}2/255$), the gradient in the larger SDXL latent is spatially sparse, so most pixels remain unperturbed and the effective MSE stays small, consistent with the lower FAR. PP-Mark$_{\mathrm{accept}}$ remains at 0.00 by construction.

\begin{figure}[H]
    \centering
    \includegraphics[width=0.6\linewidth]{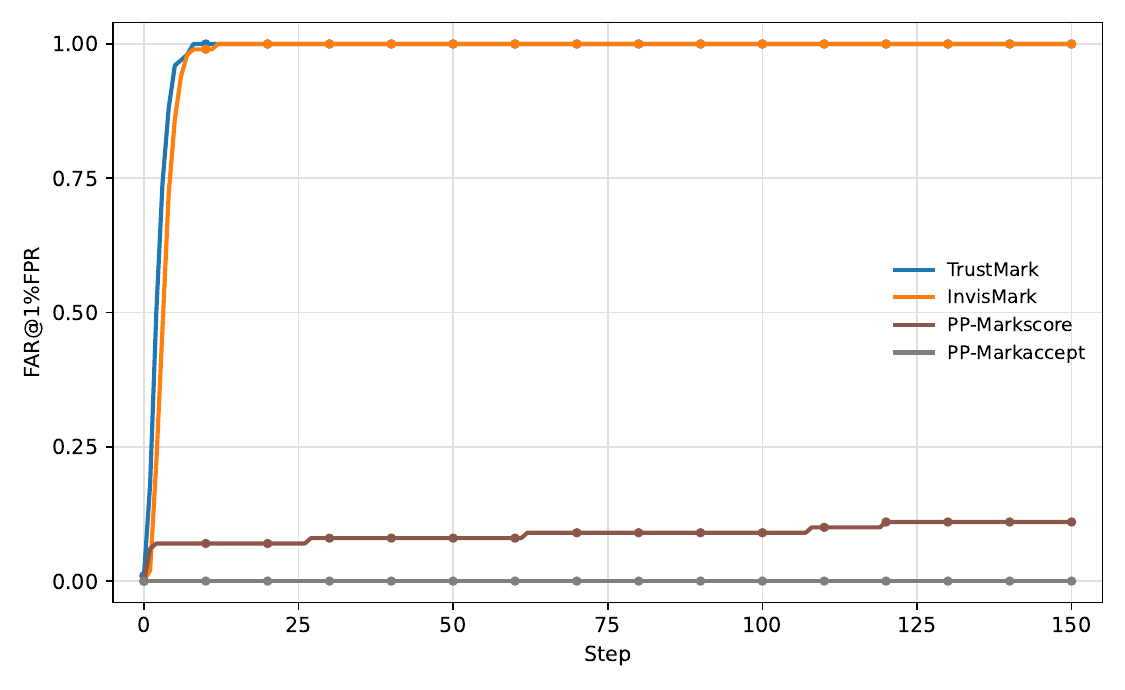}
    \caption{SDXL cumulative FAR@1\%FPR vs.\ optimization step under white-box PGD ($\epsilon{=}2/255$, $T{=}150$, $n{=}100$).}
    \label{fig:wb_forge_far_curves_sdxl}
\end{figure}

\paragraph{Summary.}
On the evaluated SDXL samples, accept mode maintains FAR$=$0.00 under both imprint transfer and white-box PGD attacks, and score mode has white-box FAR$=$0.11 compared with 0.76 on SD2.1. These results cover two models sharing a latent-diffusion/DDIM-inversion interface; generalization to substantially different generative architectures remains untested.

  \begin{table}[H]
    \caption{Post-Regeneration attack detection rate at FPR=1\% thresholds; each noise step uses 100 images; noise step is 30/60/100.}
    \label{tab:zhao-regen-full100}
    \begin{center}
      \begin{small}
        \begin{sc}
          \begin{tabular}{
          l
          S[table-format=1.2, round-precision=2]
          S[table-format=1.2, round-precision=2]
          S[table-format=1.2, round-precision=2]
          }
            \toprule
            & \multicolumn{3}{c}{\textbf{Noise step}} \\
            \cmidrule(lr){2-4}
            \textbf{Method} & \multicolumn{1}{c}{\textbf{30steps}} & \multicolumn{1}{c}{\textbf{60steps}}
              & \multicolumn{1}{c}{\textbf{100steps}} \\
            \midrule
            StableSig & 0.04 & 0.02 & 0.03 \\
            HiDDeN & 0.10 & 0.06 & 0.04 \\
            TrustMark & 0.85 & 0.63 & 0.39 \\
            InvisMark & 0.01 & 0.05 & 0.02 \\
            Tree-Ring & 0.95 & 0.93 & 0.93 \\
            PP-Mark & 0.98 & 0.98 & 0.98 \\
            \bottomrule
          \end{tabular}
        \end{sc}
      \end{small}
    \end{center}
    \vskip -0.1in
  \end{table}

\section{Regeneration Attack}\label{sec:Regeneration Attack}
We evaluate detection rates after regeneration at each noise level to measure how well the watermark survives re-synthesis. Table~\ref{tab:zhao-regen-full100} reports detection rates after regeneration.

\begin{table}[H]
  \centering
  \caption{Steganalysis removal summary (TPR@1\%FPR) by K.}
  \label{tab:steg-removal-summary}
  {\small
    \setlength{\tabcolsep}{10pt}%
    \renewcommand{\arraystretch}{1.05}%
    \begin{tabular}{lcccc}
      \toprule
      & \multicolumn{4}{c}{\textbf{K}} \\
      \cmidrule(lr){2-5}
      \textbf{Method} & \textbf{10} & \textbf{100} & \textbf{500} & \textbf{1000} \\
      \midrule
      StableSig & 0.200 & 0.200 & 0.173 & 0.171 \\
      Tree-Ring & 0.190 & 0.020 & 0.038 & 0.035 \\
      TrustMark & 0.990 & 0.960 & 0.960 & 0.960 \\
      InvisMark & 1.000 & 1.000 & 1.000 & 1.000 \\
      WIND      & 1.000 & 1.000 & 1.000 & 1.000 \\
      PP-Mark   & 0.864 & 0.959 & 0.963 & 0.958 \\
      \bottomrule
    \end{tabular}
  }
\end{table}

\section{Steganalysis Attack} \label{Steganalysis Attack}
We sweep K=\{10, 100, 500, 1000\} and evaluate the post-attack detection rate at the calibrated FPR=1\% threshold. Table~\ref{tab:steg-removal-summary} shows that StableSignature and Tree-Ring detection drops sharply as K increases, indicating increasing removal effectiveness. In contrast, TrustMark, InvisMark, WIND, and PP-Mark remain highly detectable even at large K, because their watermark patterns are content-adaptive rather than fixed, making the averaged residual a poor estimate of any individual image's watermark.
Notably, for all content-adaptive methods the averaged residual converges to zero as K grows, confirming that steganalysis is fundamentally ineffective against this encoder class. The minor TPR variation across these methods (e.g.\ PP-Mark 0.96 vs.\ InvisMark 1.00 at $K{\ge}100$) reflects differences in score-to-threshold margin rather than actual robustness: PP-Mark's tighter calibration leaves less headroom for the residual noise added at small K, whereas InvisMark's wider margin absorbs it entirely.

\section{SPSA-based Black-Box Attacks}
We evaluate PP-Mark under gradient-free SPSA \cite{spall1992spsa} attacks in both removal and forgery settings.

\paragraph{Removal.}
Given a watermarked image and its binding metadata, the adversary minimizes the verifier's score via SPSA over 50 steps. The score drops below the calibrated threshold $\tau$ early, but visual quality collapses immediately (PSNR $\approx 4.7$--$6.9$). Under any reasonable quality constraint (PSNR $\ge 30$), no valid outputs are obtained. Figure~\ref{fig:imprint-removal-step5} shows the score and PSNR traces.

\begin{figure}[H]
    \centering
    \includegraphics[width=0.6\linewidth, trim=0 0 0 25pt, clip]{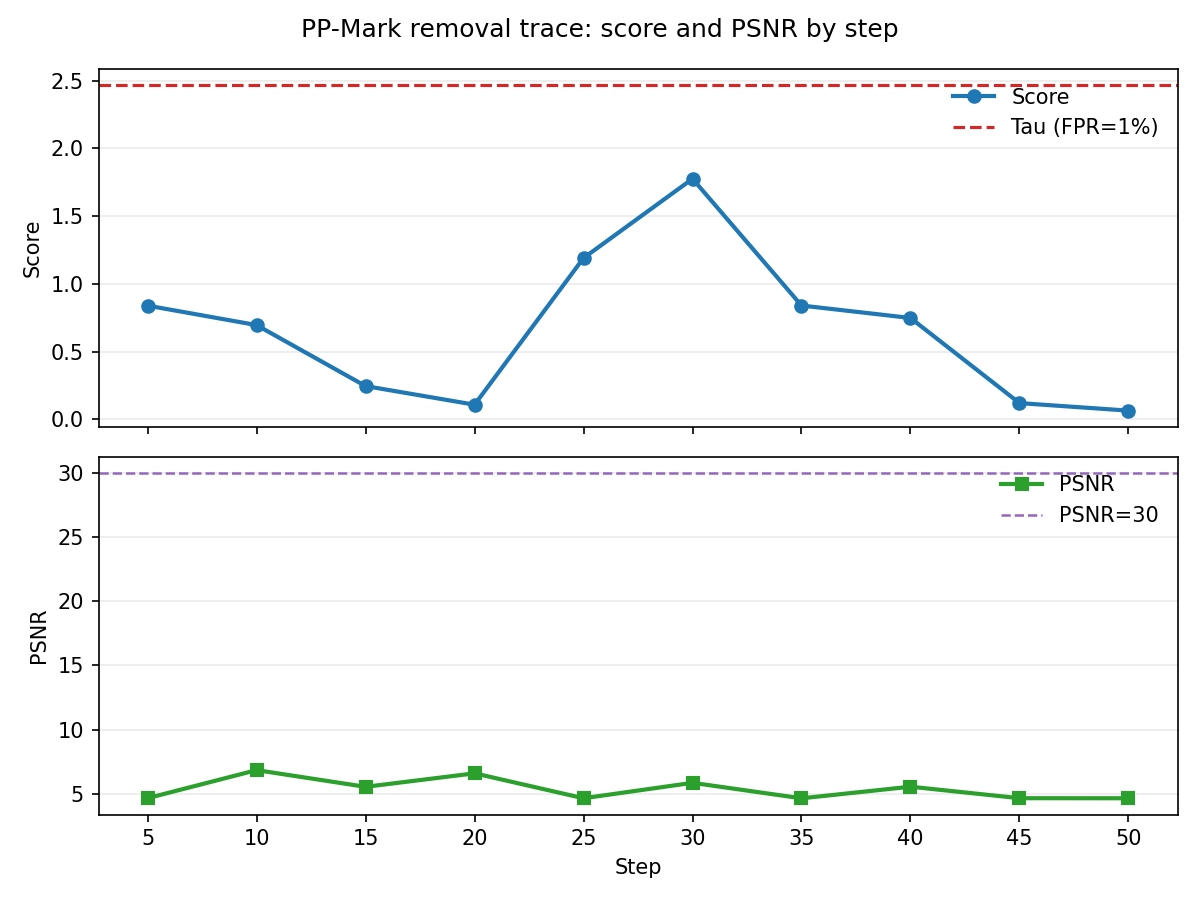}
    \caption{SPSA removal trace. Top: score vs step with $\tau$ (FPR=1\%); bottom: PSNR vs step.}
    \label{fig:imprint-removal-step5}
\end{figure}

\begin{table}[H]
    \caption{SPSA-based forgery attack on PP-Mark.}
    \label{tab:adaptive-instance}
    \centering
    {\ppmarksingletablesetup
          \begin{tabular}{
          c
          S[table-format=1.3, round-mode=none] S[table-format=1.2, round-precision=2]
          }
            \toprule
            \textbf{Step}
            & {\textbf{FAR@1\%FPR}} & {\textbf{PSNR$\uparrow$}} \\
            \midrule
            50 & 0.02 & 6.55 \\
            100 & 0.00 & 6.61 \\
            150 & 0.01 & 6.63 \\
            Avg. & 0.01 & 6.60 \\
            \bottomrule
          \end{tabular}
    \par}
\end{table}

\paragraph{Forgery.}
The adversary maximizes the verifier's score against a target binding via SPSA, using DDIM inversion to initialize the latent. Table~\ref{tab:adaptive-instance} summarizes the results: targeted FAR remains at 0.01 on average, and the optimization merely introduces destructive distortion (PSNR $\approx 6.6$).

\begin{table}[H]
\centering
\caption{Post-transformation detection rate (TPR at FPR${\leq}$1\%).
Individual transformations and the sequential pipeline are evaluated in score mode.} 
\label{tab:image-transforms}
{\ppmarksingletablesetup
\begin{tabular}{lcc}
\toprule  
Transform & SD\,2.1 & SDXL \\
\midrule
JPEG compression ($q{=}25$) & 0.94 & 1.00 \\
Gaussian blur ($8{\times}8$) & 1.00 & 1.00 \\
Additive noise ($\sigma{=}0.1$) & 0.86 & 0.95 \\
Brightness jitter (${\times}0.6$) & 0.98 & 1.00 \\
\midrule
Random crop + resize (75\%) & 0.87 & 0.74 \\
Rotation $75^\circ$ & 0.85 & 0.97 \\
\midrule
75\% crop $\to$ resize $\to$ JPEG ($q{=}25$) & $0.80^{\dagger}$ & --- \\
\bottomrule
\end{tabular}
\par}

\vspace{0.3em}
{\small\raggedright
$^{\dagger}$Success requires detection after both crop--resize and subsequent JPEG (80/100). We use the same 100 SD\,2.1 images, bicubic resizing, no survivor filtering, and the submitted FP32 DDIM-100 alignment with a frozen threshold. ---: not evaluated.
\par}
\end{table}

\section{Image Transformations} \label{Image Transformations}

We evaluate watermark robustness under individual and sequential image transformations
and report the \emph{score-mode} detection rate (TPR at FPR${\leq}$1\%).
Table~\ref{tab:image-transforms} summarizes results.

\textbf{Pixel-level transforms.}\enspace
PP-Mark maintains high detection under lossy compression, blur, noise, and brightness changes.
Even aggressive JPEG compression ($q{=}25$) preserves 94\% detection, and Gaussian blur achieves perfect retention.

\textbf{Crop and resize.}\enspace
Random cropping to 75\% of the original area followed by resizing to the original dimensions
is handled by a two-stage alignment combining integer-shift search and sub-pixel refinement.
This recovers 87\% detection on SD\,2.1 ($64{\times}64$ latent) and 0.74 on SDXL
($128{\times}128$ latent, $n{=}50$);
the VAE maps overlapping spatial regions to correlated latent coefficients,
and the refinement grid covers the residual sub-pixel misalignment.

\textbf{Rotation.}\enspace
Large rotations are recovered through a dedicated rotation-sync search over $\pm 90^\circ$.
Rotation at $75^\circ$ retains 0.85 detection on SD\,2.1
and 0.97 on SDXL; the larger $128{\times}128$ latent of SDXL
provides additional redundancy that favours sync recovery.
This single angle is reported as a representative worst-case large rotation;
the same procedure applies uniformly across angles.

\begin{figure}[H]
    \centering
    \begin{subfigure}[t]{0.36\linewidth}
      \centering
      \includegraphics[width=\linewidth]{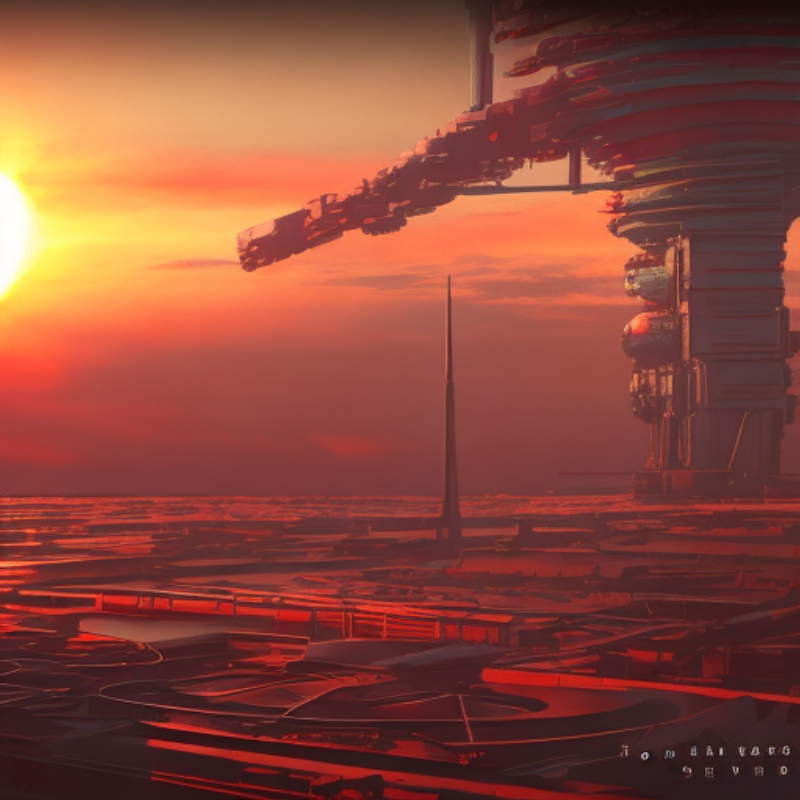}
    \end{subfigure}\hspace{0.01\linewidth}%
    \begin{subfigure}[t]{0.36\linewidth}
      \centering
      \includegraphics[width=\linewidth]{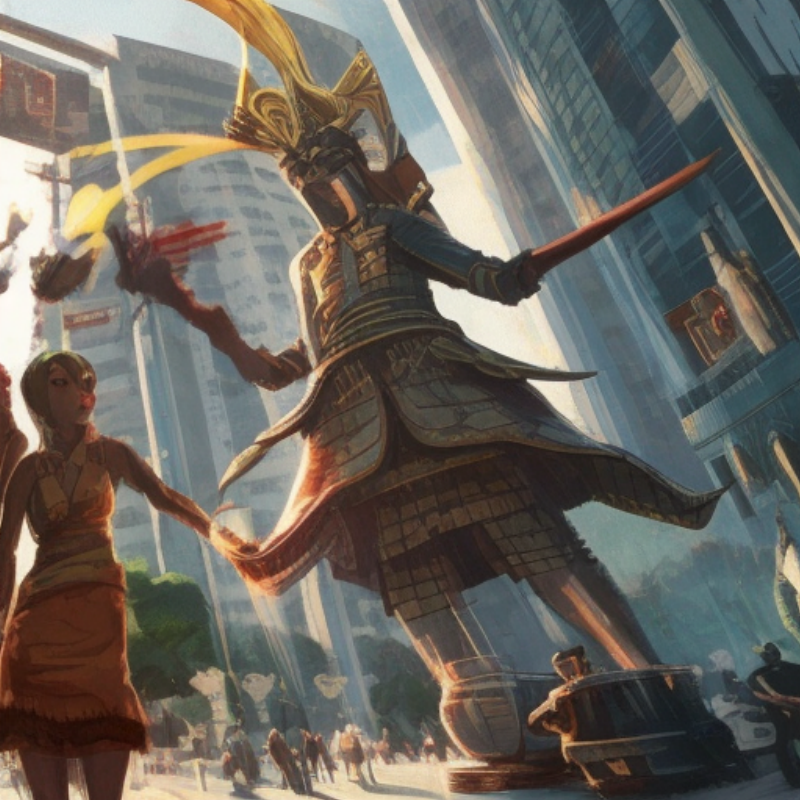}
    \end{subfigure}

    \vspace{0.2em}

    \begin{subfigure}[t]{0.36\linewidth}
      \centering
      \includegraphics[width=\linewidth]{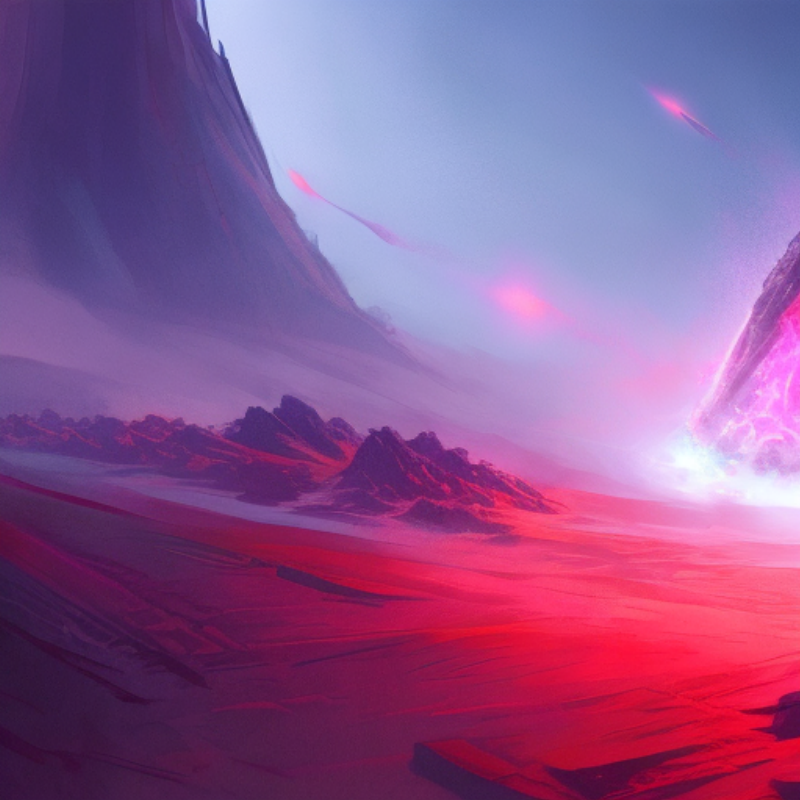}
    \end{subfigure}\hspace{0.01\linewidth}%
    \begin{subfigure}[t]{0.36\linewidth}
      \centering
      \includegraphics[width=\linewidth]{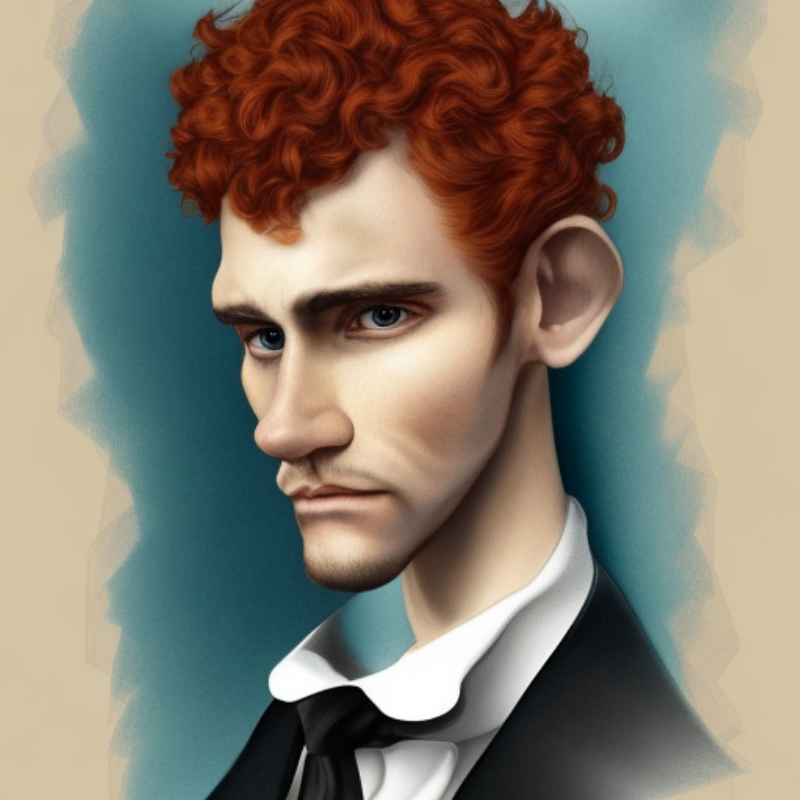}
    \end{subfigure}

    \caption{Examples of PP-Mark watermarked images}
    \label{fig:alpha4-qual}
\end{figure}

\newpage
\section{Image Quality and Distributional Fidelity}\label{Generation Distributional Fidelity}\label{app:quality}

Table~\ref{tab:quality-comparison} reports BRISQUE p95 for no-reference image quality, KID \cite{binkowski2018mmdgan} for distributional fidelity, and mean CLIP image--prompt cosine similarity (ViT-B/32) for semantic alignment. We use KID rather than FID because the 2,048-dimensional FID covariance is rank-deficient at this sample size.

\begin{table}[H]
\centering
\caption{Image quality comparison. KID is mean $\pm$ SD; arrows indicate preferred directions.}
\label{tab:kid-500}\label{tab:quality-comparison}
{\ppmarksingletablesetup
\begin{tabular}{lccc}
\toprule
Method & BRISQUE p95 $\downarrow$ & KID $\times 10^3$ $\downarrow$ & CLIP $\times 100$ $\uparrow$ \\
\midrule
PP-Mark$^{\dagger}$ & 34.14 & $1.101 \pm 0.202$ & 31.64 \\
Gaussian Shading & 27.35 & $0.621 \pm 1.096$ & 32.42 \\
HiDDeN & 105.86 & $7.239 \pm 1.474$ & 31.30 \\
InvisMark & 31.36 & $-0.900 \pm 0.895$ & 32.53 \\
RingID & 34.64 & $-0.989 \pm 0.701$ & 32.41 \\
Stable Signature & 32.69 & $-0.814 \pm 0.878$ & 32.50 \\
Tree-Ring & 41.90 & $-0.090 \pm 0.922$ & 32.30 \\
WIND & 27.42 & $-0.386 \pm 0.879$ & 32.28 \\
TrustMark & 29.44 & $-0.174 \pm 1.001$ & 32.54 \\
\bottomrule
\end{tabular}
\par}
\vspace{0.3em}
{\small\raggedright
$^{\dagger}$PP-Mark retains the submitted BRISQUE ($n{=}100$, 25 prompts $\times$ 4 seeds) and KID ($n{=}1000$, subset 500, 20 splits). Baseline BRISQUE/KID and all CLIP results use the common 100 prompt--seed pairs. Baseline KID uses a disjoint pool of 900 clean reference images; its SD is across 1,000 draws of 100 reference images, with the 100 method images fixed.
\par}
\end{table}

Small negative unbiased KID estimates can occur at finite sample sizes and do not denote negative distances. PP-Mark's submitted KID remains low, but the differing sampling protocols do not support fine cross-method KID rankings. PP-Mark does not attain the lowest BRISQUE or highest CLIP similarity; this quality profile accompanies the forgery-resistance results in Sec.~\ref{sec:imprint-forgery}.

\newpage
\section{Limitations and Discussions}\label{Additional Limitations and Discussions}

\textbf{Generative-model scope.}\enspace
Our experiments cover SD2.1 and SDXL, which share a latent-diffusion/DDIM-inversion interface. The proof circuit verifies the committed embedding trace without re-simulating the generator, while end-to-end score verification requires a compatible initial-latent embedding and recovery interface. Generalization to substantially different generative architectures remains untested.

\textbf{ZKP Overhead.}\enspace
The latest rebuttal implementation (SP1 v6.3.1, $N=1000$) produces a 2.81\,MB Core receipt and reports approximately 17.0\,s projected fresh-process receipt issuance (Table~\ref{tab:proof-backend-comparison}). The original $N$ sweeps are retained in Tables~\ref{tab:sp1-k-sweep-timing} and~\ref{tab:sp1-receipt-size}; Table~\ref{tab:sp1-k-sweep-timing} also gives a conditional, unmeasured estimate of standalone receipt verification.
This is tolerable for archival or on-request verification but costly for per-upload streaming pipelines; a STARK-to-SNARK wrapper or proof aggregation are natural reductions that we leave to future work.

\textbf{Random Crop on SDXL.}\enspace
Across geometric transformations, random crop with resize on SDXL is the single transformation where detection drops meaningfully (0.74 on SDXL vs.\ 0.87 on SD\,2.1 under the same two-stage alignment; see Table~\ref{tab:image-transforms}).
The $128{\times}128$ latent of SDXL, four times the dimensionality of SD\,2.1, amplifies sub-pixel misalignment introduced by cropping and resampling, so residual near-miss scores cluster just below $\tau$.
Refining the sub-pixel alignment grid or incorporating a learned alignment model that undoes VAE-crop interactions would close this remaining gap; we leave this to future work.

\textbf{Reliance on Explicit Proof Artifacts.}\enspace
Traditional diffusion watermarks delegate trust to a proprietary verifier; PP-Mark removes this centralized dependency to enable public auditing, but in exchange requires an explicit proof artifact to accompany the image.
If this artifact is detached, deleted, or corrupted, the image reverts to an unverified state, a challenge inherent to certificate-based provenance frameworks.
PP-Mark\textsubscript{accept} is designed to certify bit-exact originals; it correctly rejects modified images, limiting its applicability to scenarios where the original is preserved (e.g., legal evidence, regulatory audit).
For general distribution environments where images undergo lossy transforms, PP-Mark\textsubscript{score} provides statistical verification via soft correlation over the recovered latent, retaining high detection after common transforms (e.g., JPEG $q{=}25$: 94\%, blur: 100\%; see Table~\ref{tab:image-transforms}).
Future work should explore anchoring proof artifacts in content-addressed decentralized storage and on-chain registries to ensure persistence independent of any single server.

\textbf{Key Lifecycle Operational Guidance.}\enspace
PP-Mark's unforgeability is contingent on the producer's binding key $k$ remaining secret; if $k$ leaks, the adversary can mint valid receipts under the victim's identity. We treat this as an operational security requirement and recommend the following lifecycle, consistent with standard PKI practice: (i)~\emph{rotation}, each producer maintains a small set of active binding keys $\{k^{(t)}\}$ indexed by a rotation epoch $t$; metadata includes the current key identifier, so verifiers can tell which epoch a receipt belongs to; (ii)~\emph{revocation}, a compromised $k^{(t)}$ is published on a public revocation list (analogous to CRL/OCSP), and verifiers reject receipts bound to revoked epochs; (iii)~\emph{forward binding}, because $b$ only depends on the \emph{context} at generation time, rotated keys do not invalidate legitimate historical receipts as long as the corresponding epoch remains non-revoked. Integration with C2PA Content Credentials is straightforward: the epoch ID, rotation policy, and revocation URL can be carried as standard assertion fields.

\newpage
\input{checklist.tex}

\end{document}

%% file: checklist.tex
\section*{NeurIPS Paper Checklist}

\begin{enumerate}

\item {\bf Claims}
    \item[] Question: Do the main claims made in the abstract and introduction accurately reflect the paper's contributions and scope?
    \item[] Answer: \answerYes{}
    \item[] Justification: Contributions are stated in the abstract and Introduction. The conditional security scope is specified in Section~\ref{sec:formal-security} and Appendix~\ref{app:full-proof}; evaluated settings are reported in Sections~\ref{sec:sig-vs-zkp}--\ref{sec:white-box}.

\item {\bf Limitations}
    \item[] Question: Does the paper discuss the limitations of the work performed by the authors?
    \item[] Answer: \answerYes{}
    \item[] Justification: Discussed in Appendix~\ref{Additional Limitations and Discussions}.

\item {\bf Theory assumptions and proofs}
    \item[] Question: For each theoretical result, does the paper provide the full set of assumptions and a complete (and correct) proof?
    \item[] Answer: \answerYes{}
    \item[] Justification: Assumption~\ref{assum:noise} is stated in Section~\ref{sec:verification} and empirically assessed in Appendix~\ref{app:assumption-verify}. Assumption~\ref{asm:crypto} and the fixed-candidate sampling conditions accompany Theorem~\ref{thm:unforgeability} in Section~\ref{sec:formal-security}; Appendix~\ref{app:full-proof} gives the corresponding proof and scope.

    \item {\bf Experimental result reproducibility}
    \item[] Question: Does the paper fully disclose all the information needed to reproduce the main experimental results of the paper to the extent that it affects the main claims and/or conclusions of the paper (regardless of whether the code and data are provided or not)?
    \item[] Answer: \answerYes{}
    \item[] Justification: Generation, calibration, and evaluation settings are described in Section~4 and the experimental appendices. The GitHub release will provide code, configurations, and reproduction instructions identifying available inputs, external dependencies, and limitations of the retained historical data.

\item {\bf Open access to data and code}
    \item[] Question: Does the paper provide open access to the data and code, with sufficient instructions to faithfully reproduce the main experimental results, as described in supplemental material?
    \item[] Answer: \answerYes{}
    \item[] Justification: Code was supplied as supplementary material; the updated code and reproduction instructions will be released on GitHub for the camera-ready version.

\item {\bf Experimental setting/details}
    \item[] Question: Does the paper specify all the training and test details (e.g., data splits, hyperparameters, how they were chosen, type of optimizer) necessary to understand the results?
    \item[] Answer: \answerYes{}
    \item[] Justification: PP-Mark uses no learned parameters; calibration and per-attack settings are documented in Section~4 and the appendices.

\item {\bf Experiment statistical significance}
    \item[] Question: Does the paper report error bars suitably and correctly defined or other appropriate information about the statistical significance of the experiments?
    \item[] Answer: \answerYes{}
    \item[] Justification: FAR confidence intervals use the reported sample counts: individual settings use $n=100$, while Table~\ref{tab:transfer-summary} pools 100 pairs over three step budgets ($300$ outcomes per source). Cross-method comparisons there use one-sided Fisher exact tests; Appendix~\ref{app:quality} specifies the sampling protocols and SDs for KID.

\item {\bf Experiments compute resources}
    \item[] Question: For each experiment, does the paper provide sufficient information on the computer resources (type of compute workers, memory, time of execution) needed to reproduce the experiments?
    \item[] Answer: \answerYes{}
    \item[] Justification: Submitted hardware is described in Appendix~\ref{app:compute}; per-method runtimes are in Appendix~\ref{app:runtime}. The original SP1 timing sweep is in Appendix~\ref{app:sample-size} (Table~\ref{tab:sp1-k-sweep-timing}); sizes and the rebuttal backend comparison are in Appendix~\ref{app:proof-size}. Updated rows distinguish projected issuance time, estimated receipt-verification time, and measured artifact size, and identify unreported quantities.
    
\item {\bf Code of ethics}
    \item[] Question: Does the research conducted in the paper conform, in every respect, with the NeurIPS Code of Ethics \url{https://neurips.cc/public/EthicsGuidelines}?
    \item[] Answer: \answerYes{}
    \item[] Justification: Conforms with the NeurIPS Code of Ethics; no human subjects or sensitive data.

\item {\bf Broader impacts}
    \item[] Question: Does the paper discuss both potential positive societal impacts and negative societal impacts of the work performed?
    \item[] Answer: \answerYes{}
    \item[] Justification: Discussed in Section~6.
    
\item {\bf Safeguards}
    \item[] Question: Does the paper describe safeguards that have been put in place for responsible release of data or models that have a high risk for misuse (e.g., pre-trained language models, image generators, or scraped datasets)?
    \item[] Answer: \answerNA{}
    \item[] Justification: No new generative models, datasets, or scraped data are released.

\item {\bf Licenses for existing assets}
    \item[] Question: Are the creators or original owners of assets (e.g., code, data, models), used in the paper, properly credited and are the license and terms of use explicitly mentioned and properly respected?
    \item[] Answer: \answerYes{}
    \item[] Justification: All baselines and the SP1 zkVM are cited; pretrained Stable Diffusion 2.1 and SDXL 1.0 weights are used under their public Hugging Face licenses.

\item {\bf New assets}
    \item[] Question: Are new assets introduced in the paper well documented and is the documentation provided alongside the assets?
    \item[] Answer: \answerYes{}
    \item[] Justification: The updated PP-Mark codebase, documentation, configurations, and reproduction instructions are being prepared for the camera-ready GitHub release.

\item {\bf Crowdsourcing and research with human subjects}
    \item[] Question: For crowdsourcing experiments and research with human subjects, does the paper include the full text of instructions given to participants and screenshots, if applicable, as well as details about compensation (if any)?
    \item[] Answer: \answerNA{}
    \item[] Justification: No crowdsourcing or human subjects.

\item {\bf Institutional review board (IRB) approvals or equivalent for research with human subjects}
    \item[] Question: Does the paper describe potential risks incurred by study participants, whether such risks were disclosed to the subjects, and whether Institutional Review Board (IRB) approvals (or an equivalent approval/review based on the requirements of your country or institution) were obtained?
    \item[] Answer: \answerNA{}
    \item[] Justification: No human subjects.

\item {\bf Declaration of LLM usage}
    \item[] Question: Does the paper describe the usage of LLMs if it is an important, original, or non-standard component of the core methods in this research? Note that if the LLM is used only for writing, editing, or formatting purposes and does \emph{not} impact the core methodology, scientific rigor, or originality of the research, declaration is not required.
    \item[] Answer: \answerNA{}
    \item[] Justification: Core methodology does not use LLMs.

\end{enumerate}